\documentclass[journal]{IEEEtran}%

\usepackage{tex/signalpreamble}
\usepackage{tikz-cd}
\usepackage{pgfplots}
\pgfplotsset{compat=1.5}%%%%%%%%%%%%%%%%%%%

\crefname{question}{question}{questions}

\theoremstyle{definition}

\newacronym[plural=GDSs]{GDS}{GDS}{graph distribution-valued signal}
\newacronym[plural=SAGSs]{SAGS}{SAGS}{signal adaptive graph structure}
\newacronym{GSP}{GSP}{graph signal processing}
\newacronym{GFT}{GFT}{graph Fourier transform}
\newacronym{GSO}{GSO}{graph shift operator}
\newacronym{GDS-FT}{GDS-FT}{GDS Fourier transform}
\newacronym{CDF}{cdf}{cumulative distribution function}
\newacronym{GMM}{GMM}{Gaussian mixture model}

\myexternaldocument{supplementary}

\begin{document}
\date{}
\title{Graph Distribution-valued Signals in Wasserstein Spaces: Theory and Applications}
\author{Yanan Zhao, Feng Ji, Xingchao Jian and Wee Peng Tay,~\IEEEmembership{Senior~Member,~IEEE} %
%\thanks{This research is supported by the Singapore Ministry of Education Academic Research Fund Tier 2 grant MOE2018-T2-2-019 and A*STAR under its RIE2020 Advanced Manufacturing and Engineering (AME) Industry Alignment Fund – Pre Positioning (IAF-PP) (Grant No. A19D6a0053).}% 
\thanks{The authors are with the School of Electrical and Electronic Engineering, Nanyang Technological University, 639798, Singapore. (e-mail: yanan002, jifeng, xingchao001, wptay@ntu.edu.sg).}%
}

\markboth{submitted to IEEE TRANSACTIONS ON SIGNAL PROCESSING}%
{How to Use the IEEEtran \LaTeX \ Template}
%{Tay: xxxx}
\maketitle

\begin{abstract}
We introduce a framework for \gls{GSP} in which signals are represented as \glspl{GDS}, i.e., probability measures in a Wasserstein space. This perspective addresses fundamental limitations of classical vector-based \gls{GSP}, including the requirement for complete synchronous observations across vertices and the need for strict temporal correspondence in observed filter input--output pairs. Furthermore, by modeling the graph structure as a distribution conditioned on signal realizations, we provide a principled approach to signal-dependent graph structures, which are common in real-world applications, while explicitly encoding uncertainty in graph topology. Our framework inherently captures uncertainty and stochasticity while strictly generalizing traditional graph signals, which can be interpreted as Dirac delta measures. We develop a systematic correspondence between foundational GSP concepts and their GDS analogs, showing that classical formulations emerge as special cases of our framework. We establish theoretical continuity results for GDS transforms, providing stability guarantees for input perturbations and distribution approximations.
We demonstrate the utility of this approach through example applications, including graph filter learning and anomaly detection, and validate its effectiveness through empirical studies.
\end{abstract}

\glsresetall


\section{Introduction} \label{sec:intro}

Graphs provide a powerful framework for modeling complex systems in diverse domains, including social networks, transportation systems, and sensor networks. In the classical \gls{GSP} framework \cite{San14, Shu13, Ort18, Girault2018, Ji19, JiaJiTay23}, graph signals are represented as vectors, where each entry corresponds to the value at a specific node. This representation enables the application of linear operators such as the \gls{GFT}, convolution, and graph filters \cite{Shu13, Ort18, ortega_2022} to analyze signal structure and inter-node relationships.

Despite its widespread adoption, the vector-based \gls{GSP} framework exhibits several critical limitations:
(i) \emph{Assumption of complete observations;} Classical \gls{GSP} frameworks assume that signals are observed synchronously and completely across all nodes, so that operators such as the \gls{GFT} can be applied to the full signal vector. However, in practice, data collection is often asynchronous or incomplete \cite{Hua20, Ji19, Ji25}, making this assumption unrealistic.
(ii) \emph{Requirement for strict signal correspondence.} While statistical \gls{GSP} methods \cite{PerVan:J17,MarSegLeuRib:J17,JiaTayEld24,KroRouEld:J22,JiaTay22,SagRou23} incorporate randomness via random variables, graph filters are still defined as vector-to-vector mappings and require paired input--output signals, e.g., $\set{(\bx_i, \by_i)}_{i=1}^m$, with the filter $\bF$ trained to satisfy $\bF(\bx_i) \approx \by_i$. This imposes a rigid one-to-one correspondence between input and output signals. In practice, such alignment is often imperfect due to temporal shifts, overlaps, periodic patterns, or missing and shuffled data, thereby limiting the applicability of these methods.
(iii) \emph{Assumption of a deterministic graph topology.} Classical \gls{GSP} represents graph signals as vectors indexed by the vertices of a \emph{fixed} graph, and defines the \gls{GFT} via the eigendecomposition of a predetermined \gls{GSO} (e.g., the adjacency matrix, Laplacian matrix, or their normalized variants). This formulation inherently assumes that the graph topology is fully known and deterministic. In many applications, however, the graph structure may be uncertain or partially observed \cite{stefania2019,dong2016,JiTayOrt23,Ji23e}, posing significant challenges to this classical framework.

To address these challenges, we model both graph signals and graph topologies as probability distributions rather than as random vectors or sample realizations. In contrast to statistical \gls{GSP}, which typically operates either on samples with prior information like stationarity \cite{KroRouEld:J22,SagRou23,PerVan:J17,MarSegLeuRib:J17}, our framework treats the full probability distribution as the fundamental signal object. Accordingly, filters are optimized directly in distribution space, rather than in the node-value vector space. When the signal and topology distributions are unknown, they are first estimated from samples, and these estimates are then used for filter learning. This differs from statistical \gls{GSP}, where filters are learned directly from samples. Most statistical \gls{GSP} methods therefore remain subject to the limitations discussed above, because they continue to rely on vector-based signal representations. In our framework, the classical vector space of graph signals is replaced by the Wasserstein space \cite{Vil09,Kol56,Dur19}, whose elements are probability measures on $\mathbb{R}^{N}$; we call such a signal a \textbf{\gls{GDS}}. As a further generalization, we model the graph structure itself as a distribution conditioned on the signal realization, referred to as \textbf{\gls{SAGS}}.

\begin{figure*}[!t]
\centering

% Top: partial observations
\begin{subfigure}[b]{0.325\textwidth}
    \centering
    % \scalebox{1}[0.9]{
    \includegraphics[width=\textwidth,
        trim={2cm 8.4cm 1.4cm 7cm},
        clip
      ]{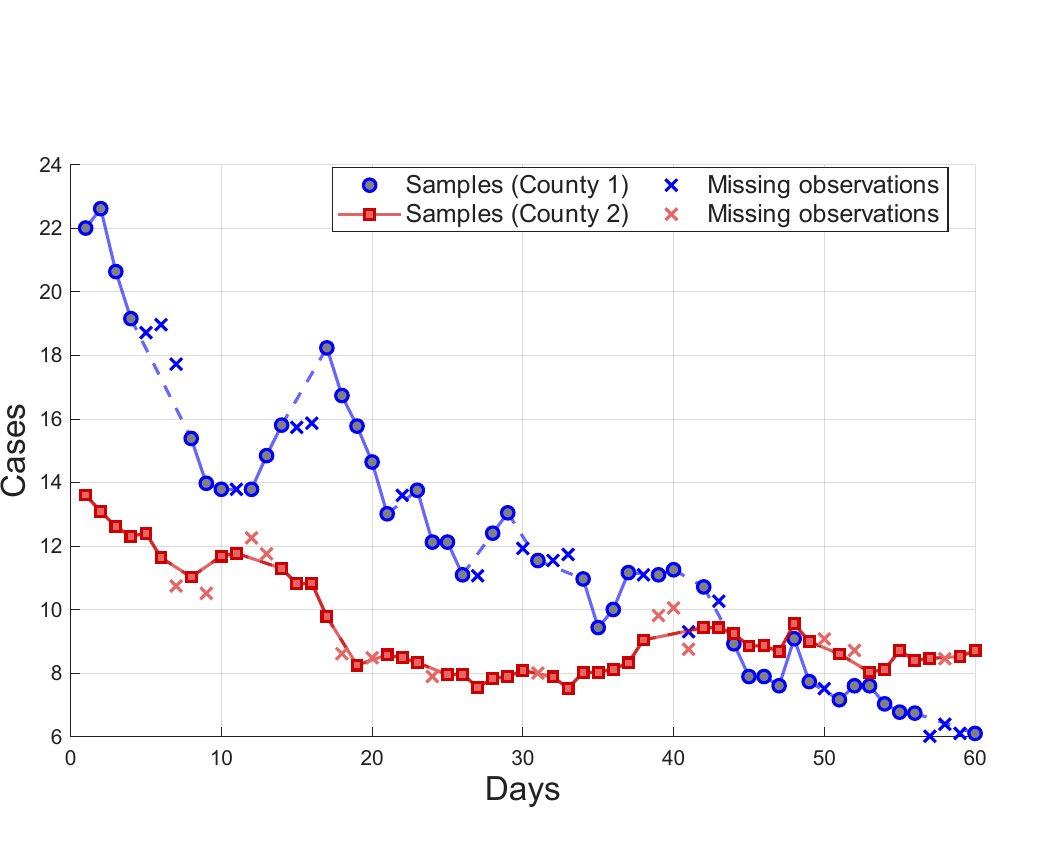}
    % }
    \vspace{-7mm}
    \caption{\scriptsize Daily cases from two counties}
    \label{fig.covid_samples_partial}
\end{subfigure}
% \par\vspace{-2mm}
% \par\vspace{2mm}
% Bottom-left: distribution from complete observations
\begin{subfigure}[b]{0.325\textwidth}
\centering
\includegraphics[width=1\textwidth, trim={1.8cm 8.4cm 1.8cm 7cm}, clip]{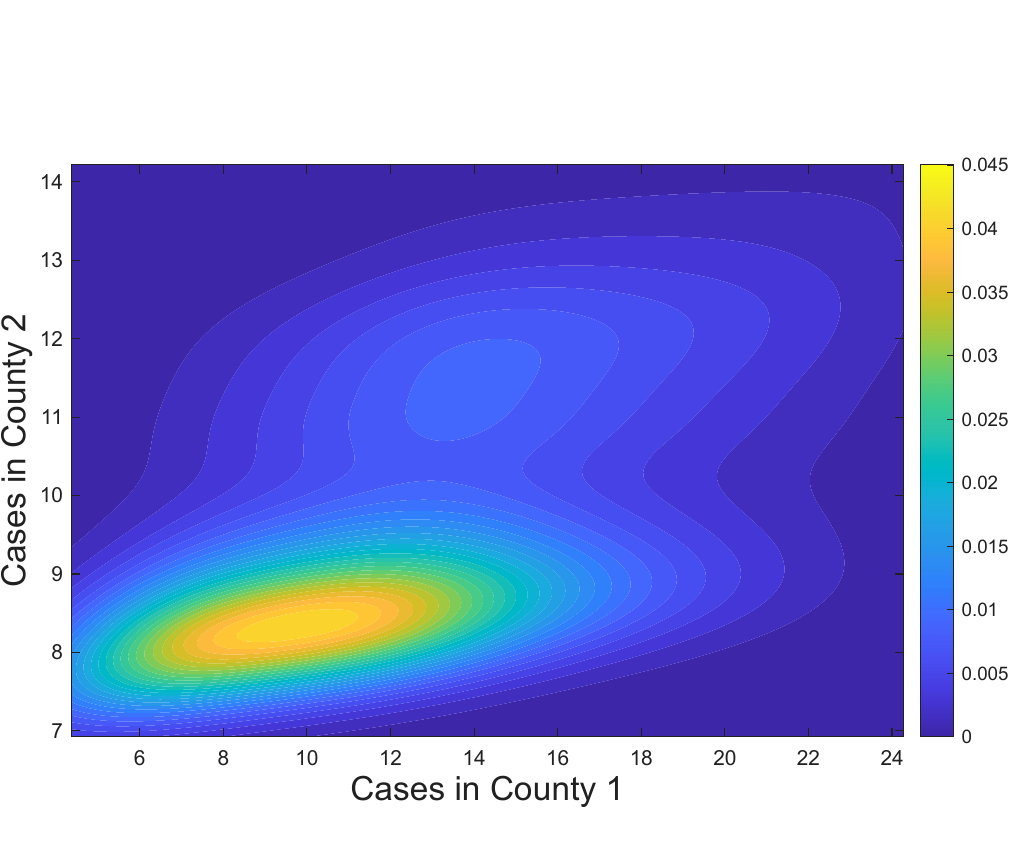}
\vspace{-7mm}
\caption{\scriptsize Joint distribution from complete observations}
\label{fig.covid_distribution_full}
\end{subfigure}
% \hfill
% Bottom-right: distribution from partial observations
\begin{subfigure}[b]{0.325\textwidth}
\centering
\includegraphics[width=1\textwidth, trim={1.8cm 8.4cm 1.8cm 7cm}, clip]{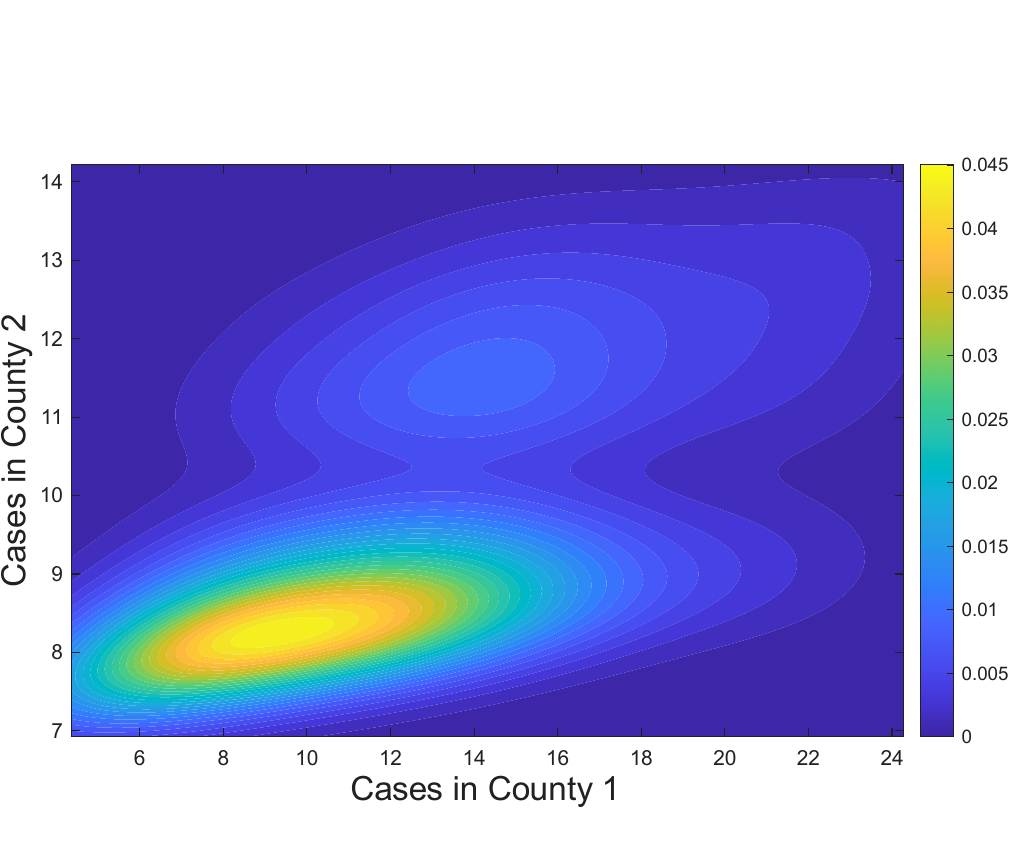}
\vspace{-7mm}
\caption{\scriptsize Joint distribution from partial observations}
\label{fig.covid_distribution_partial}
\end{subfigure}
\vspace{-2mm}
\caption{(a) COVID-19 daily cases from two counties over 60 days, together with the corresponding partial observations generated using an independent per-county missing rate of $20\%$; crosses indicate the masked entries.
(b) Joint distribution estimated from the complete observations using KDE-estimated marginals and a Gaussian copula, representing the fully observed trajectories in (a) as realizations from an underlying distribution within the GDS framework. (c) Joint distribution estimated from the partial observations using the same Gaussian-copula model.}
\label{fig:samples_vs_distributions}
% \vspace{-3mm}
\end{figure*}

To illustrate the distinction from statistical \gls{GSP}, consider COVID-19 daily case counts from 58 counties (2020--2022). In classical statistical \gls{GSP}, each day is represented as a graph signal vector $\bx \in \Real^{58}$, where the $i$-th entry is the reported count in county $i$. This representation is convenient when county-level records are complete. In practice, however, many days include missing county reports. The vector-based approach then requires either discarding incomplete days, which can cause substantial data loss, or imputing missing values, which may introduce bias into subsequent filter learning.
In contrast, the \gls{GDS} framework treats county-level observations over time as samples from an underlying joint \emph{probability distribution}, rather than as isolated daily vectors. Filtering is thus learned in the Wasserstein space of probability measures instead of the vector space of node-valued graph signals. This perspective captures both day-to-day variability and inter-county dependence. 

\Cref{fig:samples_vs_distributions} illustrates this distinction: \cref{fig.covid_samples_partial} shows 60 days of two-county case counts as vector-valued samples, whereas \cref{fig.covid_distribution_full} shows the corresponding joint distribution estimated from the same period using KDE marginals and a Gaussian copula. Under missing observations, this formulation can still leverage partial data through distribution estimation rather than requiring complete daily vectors. For an independent per-county missing rate of $20\%$, the probability that a two-dimensional daily observation is complete is $0.64$; hence only about $38$ of 60 days are directly usable as complete vectors. By contrast, each county retains about $80\%$ of its observations for marginal estimation, while pairwise-complete observations are used to estimate copula dependence. As shown in \cref{fig.covid_distribution_full,fig.covid_distribution_partial}, the Gaussian-copula estimate from partial observations remains close to the complete-observation estimate. Their discrepancy is quantified by
\begin{align*}
d_{\mathrm{TV}}\left(\hat{\mu}_{\mathrm{full}},\hat{\mu}_{\mathrm{partial}}\right)
=\frac{1}{2}\int_{\Real^{2}}\left|\hat{p}_{\mathrm{full}}(\bx)-\hat{p}_{\mathrm{partial}}(\bx)\right|\,\ud\bx,
\end{align*}
where $\hat{p}_{\mathrm{full}}$ and $\hat{p}_{\mathrm{partial}}$ denote densities estimated from complete and partial observations, respectively. The resulting value, $d_{\mathrm{TV}}=0.043$, indicates that the dominant distributional structure is preserved despite missing entries. The effect of missing observations is further examined in \cref{sec:exper_G_learning} and \cref{fig:masking}, where we compare vector-based \gls{GSP} filtering with the proposed \gls{GDS} framework.

This flexibility comes with a trade-off. Classical statistical \gls{GSP} is simpler when complete, reliable vector-valued samples are available. The \gls{GDS} framework instead requires estimation of the underlying joint distribution, including inter-county dependence, so its performance depends on distribution-model quality. Nevertheless, when missingness, sampling variability, and spatial dependence are substantial, learning filters in distribution space provides a more flexible alternative to vector-based statistical \gls{GSP}.

Our framework is a generalization of classical \gls{GSP} by replacing the underlying vector space with the Wasserstein space, so that traditional graph signals, which correspond to Dirac delta distributions, become special cases of \glspl{GDS}. This generalization enables a principled treatment of uncertainty and facilitates richer modeling of graph-structured data. Beyond generalizing the signals, we also substitute fixed graph topology with signal-adaptive graph structures, yielding a flexible framework that accommodates uncertainty in both signals and graphs. Methodologically, this shift requires moving away from linear algebra toward tools from analysis and probability theory, giving our approach the flavor of classical Fourier theory \cite{Rud87} rather than that of algebraic signal processing \cite{Pus08}.     

While this greater flexibility brings clear modeling benefits, it also introduces practical challenges. For example, the \gls{GDS} framework requires additional data to accurately estimate distributions, especially in high-dimensional settings. In practice, applications often do not require full distributional knowledge. Instead, one can assume a parametric model like a Gaussian copula to capture key statistical properties. In some other applications like node classification, deep learning models can be employed to learn a proxy distribution given by the softmax output of the model for the node labels, with the GDS framework applied to filter these node distributions for downstream tasks \cite{JiLeeMen23,JiZhaLee25}. Another challenge is increased computational complexity due to operations in the Wasserstein space. However, this can be mitigated by working with tractable families of distributions, as justified by our theory and demonstrated in our numerical experiments (cf.\ \cref{sec:experiments}).

%The central objective of this work is to construct a conceptual dictionary that systematically associates the key concepts, e.g., \gls{GFT}, from classical GSP to their counterparts in the \gls{GDS} framework. The dictionary enables the direct translation of vector-based models and algorithms into the distribution-valued setting, facilitating principled extensions of existing GSP theory. 

Our main contributions are summarized as follows: 
\begin{itemize}
\item We develop a unified \gls{GDS} framework and its generalized extension in which graph signals are modeled as probability measures in Wasserstein space and graph topologies are modeled as signal-adaptive distributions. Within this framework, we define distribution-level analogs of core \gls{GSP} operations (Fourier transform, filtering, and convolution) and provide a systematic dictionary from classical \gls{GSP} to \gls{GDS}, with vector-valued \gls{GSP} recovered as a special case.
\item We establish new analytic guarantees for generalized GDS transforms, including uniform continuity on compact signal domains and pointwise continuity under a Lipschitz condition on the induced transform distribution. These results provide a stability foundation for learning and inference from approximate distributions and perturbed inputs.
\item We illustrate the \gls{GDS} framework in two representative tasks, graph filter learning and anomaly detection, and verify through numerical experiments that the distribution-level formulation is practically effective, particularly under incomplete or misaligned observations.
\end{itemize} 

This paper builds on our preliminary work in \cite{zhao2025GDSICASSP}, which introduced the \gls{GDS} representation of graph signals as probability measures in Wasserstein space. In this extended version, we (i) generalize from fixed graph structure to signal-adaptive random graph structure, (ii) formalize a systematic operator-level correspondence between GSP and GDS, (iii) establish continuity-based stability theory for generalized transforms, and (iv) demonstrate concrete gains in graph filter learning and anomaly detection. These additions address the main limitations of the preliminary version, namely limited formal development and lack of application-level validation.
%A more concise, mathematically focused exposition, including its connection to \cite{JiTayOrt23}, is available in our preprint \cite{Ji23d}.

Because the signals of interest in this work are probability distributions in Wasserstein spaces, we formulate the \gls{GDS} framework in terms of probability measures rather than random variables. To facilitate comparison with classical \gls{GSP}, whenever a main \gls{GDS}-related definition involves a probability distribution $\mu$, we accompany it with \emph{the statement ``$\rhd\ \bx \sim \mu$''}, indicating the random variable $\bx$ distributed according to $\mu$.

The rest of this paper is organized as follows. \cref{sec:wsa} formulates the mathematical setting underlying our 
theory via the Wasserstein space. The complete \gls{GDS} framework is then developed in \cref{sec:gds_framework} and \cref{sec.general_gds}. Example applications of the GDS framework are presented in \cref{sec.G_learning} and \cref{sec:anomaly}, and numerical results in \cref{sec:experiments}, followed by concluding remarks in \cref{sec.con}. 

\emph{Notation.} Scalars and scalar-valued functions are denoted by plain lowercase letters (e.g., $x$), whereas vectors and vector-valued functions are denoted by bold lowercase letters (e.g., $\bx$). Matrices are denoted by bold uppercase letters (e.g., $\bA$), and linear operators are also written in boldface. Probability distributions are denoted by lowercase Greek letters (e.g., $\mu, \nu, \gamma$), with $\delta$ reserved for the Dirac measure. Calligraphic letters denote spaces (e.g., $\calP$), and $G$ denotes a graph throughout. The set of real numbers is denoted by $\Real$, and $M_{N}(\Real)$ denotes the space of $N \times N$ real matrices. The transpose is denoted by $\parens{}\T$, $\Tr(\cdot)$ denotes the trace operator, and $\|\cdot\|_{\mathrm{op}}$ denotes the operator norm induced by the Euclidean norm. The symbol $\circ$ denotes function composition. The pushforward of a measure $\mu$ through a function or operator $\bF$ is denoted by $(\bF)_*\mu$. Throughout this paper, all measures are Borel measures.

\section{Preliminaries: Wasserstein Spaces And Graph Distribution-Valued Signals} 
\label{sec:wsa}

In this section, we first define the Wasserstein space and then introduce the notion of \glspl{GDS} as probability measures in Wasserstein space, which generalizes traditional graph signals. We provide examples to illustrate the concept and underlying intuition.

Throughout this paper, we consider a connected, simple, and undirected graph $G=(V,\bA)$ with $|V|=N$ vertices and (weighted) edges represented by the adjacency matrix $\bA=[a_{ij}]_{i,j\in\calV}$. Its Laplacian is $\bL_{G}=\bD-\bA$, where $\bD$ is the degree matrix. We denote by $\bS_{G}$ a generic graph shift operator, such as $\bA$ or $\bL_{G}$. 

Recall that a traditional graph signal is a vector $\bx \in \bbR^N$ that assigns a real value to each vertex. From a probabilistic perspective, this can be viewed as the Dirac measure $\delta_{\bx}$. This motivates a natural generalization in which graph signals are modeled as probability measures on $\bbR^N$, residing in the Wasserstein space \cite{Vil09}. This probabilistic formulation unifies deterministic signals and their statistical descriptions, providing a foundation for the proposed \gls{GDS} framework.

\begin{Definition}[Wasserstein space]
\label{def.wasse.space}
Let $(\calX, \norm{})$ be a normed space and $\calP(\calX)$ the set of Borel probability measures on $\calX$, i.e., $\calP(\calX):=\set{\mu \given \mu~\text{is a probability measure on}~(\calX, \calB(\calX))}$, where $\calB(\calX)$ is the Borel $\sigma$-algebra. For $p \ge 1$, the Wasserstein space of order $p$ is
\begin{align}
\calP_{p}(\calX) = \set*{\mu \in \calP(\calX)\given \int_{\calX} \norm{\bx}^{p}\ud\mu(\bx) < \infty}.
\end{align}
This space is equipped with the $p$-Wassertain distance: for $\mu_1,\mu_2 \in \calP_{p}(\calX)$,
\begin{align}
W_{p}(\mu_{1},\mu_{2}):= \left(\inf_{\gamma \in \Gamma(\mu_{1},\mu_{2})}\int_{\calX \times \calX} \norm{\bx-\by}^{p}\ud \gamma(\bx,\by)\right)^{1/p},
\end{align}
where $\mu_{1}, \mu_{2}\in \calP_{p}(\calX)$, and $\Gamma(\mu_{1},\mu_{2})$ denotes the set of all couplings of $\mu_{1}$ and $\mu_{2}$, i.e., joint probability measures on $\calX\times\calX$ with marginals $\mu_{1}$ and $\mu_{2}$.
\end{Definition}

For a graph $G$, the node signal sample space is $\calX=\bbR^N$ equipped with the Euclidean norm $\norm{\bx-\by}_{2}$. The corresponding Wasserstein space $\calP_{p}(\bbR^N)$ consists of all probability measures on $\bbR^N$ with finite $p$-th moments. %i.e., $\int_{\bbR^N} \norm{x}_2^p \ud\mu(x) < \infty$.

\begin{Definition}[Graph distribution-valued signal]
\label{def.gds}
A graph distribution-valued signal (GDS) is a Borel probability measure $\mu \in \calP_{p}(\bbR^{N})$. %where $\calP_{p}(\bbR^{N})$ denotes the $p$-Wasserstein space of probability measures on $\bbR^N$, the graph signal sample space. 
The space $\calP_{p}(\bbR^{N})$ is referred to as the space of \glspl{GDS}.
\end{Definition}

The Wasserstein distance measures the minimal ``work'' required to transport one probability distribution to another, where work is the product of the mass and distance transported. Equipped with $W_{p}$, the space $\calP_{p}(\bbR^{N})$ is complete and separable \cite{Vil09}. While computing $W_p$ for general measures is challenging, closed-form solutions exist in special cases.

\begin{Example} \label{eg:imd}
\begin{enumerate}[(a)]
\item \label{it:imd} If $\mu_1 = \delta_{\bx}$ and $\mu_2 = \delta_{\by}$, then $W_{2}(\delta_\bx,\delta_\by) = d(\bx,\by)$. Thus, the space of traditional graph signals $\bbR^N$ embeds isometrically into the space of \gls{GDS}s $\calP_{2}(\bbR^N)$. 

\item Let $\mu_1= \calN(\bm{m}_1,\bSigma_1)$ and $\mu_2= \calN(\bm{m}_2,\bSigma_2)$ be two non-degenerate Gaussian distributions on $\bbR^N$ with means $\bm{m}_1,\bm{m}_2$ and covariance matrices $\bSigma_1,\bSigma_2$ respectively. Then, their 2-Wasserstein distance is given by
\begin{align}
\ml{W_{2}(\mu_1,\mu_2)^2  = \norm{\bm{m}_1-\bm{m}_2}_{2}^2 \\ 
+ \Tr\left(\bSigma_1+\bSigma_2-2\left(\bSigma_2^{1/2}\bSigma_1\bSigma_2^{1/2}\right)^{1/2}\right).}
\end{align}
Therefore, minimizing the discrepancy between two Gaussian distributions in the Wasserstein space involves matching not only their means but also their covariances.

\item Let $\mathrm{GMM}_{K}(\Real^{N})$ denote the set of \glspl{GMM} on $\Real^{N}$ with at most $K$ components; i.e., the set of measures that can be written as
$\mu = \sum_{k=1}^{K'} a_k \mu_k$, where $K'\le K$, $\ba=(a_1,a_2,\dots,a_{K'})\T$, and $\{\mu_k\}_k$ is a family of pairwise distinct Gaussian distributions (each with mean $\bm{m}_k$ and covariance matrix $\bSigma_k$). Let $\mathrm{GMM}_{\infty}(\Real^{N})$ denote the set of all finite Gaussian mixtures on $\Real^{N}$, i.e.,
$\mathrm{GMM}_{\infty}(\Real^{N}) = \bigcup_{K\ge 0} \mathrm{GMM}_{K}(\Real^{N})$.
If $\mu\in \mathrm{GMM}_{K}(\Real^{N})$ and $\nu \in \mathrm{GMM}_{L}(\Real^{N})$, the Mixture-Wasserstein distance of order $2$ \cite{Delon2020} is defined as 
\begin{align*}
\ml{ MW_{2}(\mu, \nu)^{2} \\
= \inf_{\gamma\in \Gamma(\mu,\nu) \cap \mathrm{GMM}_{\infty}(\Real^{2N})} \int_{\Real^{N}\times \Real^{N}} \|\bx-\by\|_{2}^{2} \ud\gamma(\bx,\by),}
\end{align*}
which admits the equivalent formulation \cite{Chen2018}:
\begin{align}
\label{eq.MW2}
MW_{2}(\mu, \nu)^{2}= \inf_{\omega \in \Pi(\ba,\bb)} \sum_{k,l} \omega_{kl} W_{2}(\mu_{k},\nu_{l})^{2}.
\end{align}
where $\ba = (a_{1}, a_{2},\dots,a_{K})\T$ and $\bb = (b_{1}, b_{2},\dots, b_{L})\T$ are the mixture weights, and $\Pi(\ba,\bb)$ denotes the set of admissible couplings. Since the admissible couplings in $MW_2$ form a strict subset of those in $W_2$, this metric generally upper-bounds the classical Wasserstein distance; namely, $MW_2(\mu,\nu) \ge W_2(\mu,\nu)$. In practice, $MW_2(\mu,\nu)$ is often more tractable and can serve as a proxy for $W_2(\mu,\nu)$ in optimization problems.

\end{enumerate}
\end{Example}

\section{The GDS Processing Framework}
\label{sec:gds_framework}

In this section, we introduce a signal processing framework for \glspl{GDS} that generalizes classical \gls{GSP}. We first consider the basic setting in which the graph topology is fixed, and extend core GSP concepts to their GDS counterparts.
To demonstrate the utility of the framework, we present an example application in graph filter learning.
In \cref{sec.general_gds}, we further generalize the framework by allowing the graph distribution to be signal-adaptive.

\subsection{\gls{GDS} Fourier Transform} 

Recall that in traditional \gls{GSP}, given a GSO $\bS_G = \bU_G \bLambda_G \bU_G\T$, where $\bU_G$ is the unitary matrix of eigenvectors and $\bLambda_G$ is the diagonal matrix of eigenvalues, the \gls{GFT} of a graph signal $\bx\in\bbR^N$ is given by the orthogonal basis change $\widehat{\bx}=\bU_{G}\T\bx$. We extend this notion to \glspl{GDS} as follows.

\begin{Definition}[\gls{GDS} Fourier transform]\label{def.GDS_FT}
Let $\bS_{G} = \bU_G \bLambda_G \bU_G\T$ be a GSO, and let $\mu \in \calP_p(\bbR^N)$ be a \gls{GDS} on $G$. The \emph{\gls{GDS-FT}} is the pushforward measure
\begin{align}
\hat{\mu}:= \left(\bU\T_{G}\right)_{*} \mu \in \calP_p(\mathbb{R}^N),
\end{align}
i.e., for any Borel set $B$, $\widehat{\mu}(B) = \mu\parens*{\parens*{\bU\T_{G}}^{-1}(B)}$. 
As a reminder, the symbol $\rhd$ introduces the corresponding random variable interpretation: $\rhd\ \bx \sim \mu,\ \bU_G\T\bx \sim \hat{\mu}$, which establishes the correspondence between classical GSP and our framework. 
\end{Definition}

Intuitively, \gls{GDS-FT} transports a distribution from the vertex domain to the graph frequency domain via the pushforward induced by $\bU\T_{G}$. When $\mu = \delta_{\bx}$ is a Dirac measure, we have $\widehat{\mu} = \delta_{\bU\T_{G} \bx}$, which recovers the classical \gls{GFT}. We have the following properties.
% The inverse \gls{GDS-FT} is given by $\mu = (\bU_\calG)_{\#}\widehat{\mu}$.  

\begin{Proposition}\label{thm:properties_fixed}
Let $G$ be fixed graph with shift operator $\bS_{G}$ and eigenbasis $\bU_{G}$. Then the following properties hold. 
\begin{enumerate}[(i)]
\item \label{item:well-definedness}
(Well-definedness). The GDS-FT is well-defined, i.e., for any $\mu \in \calP_p(\bbR^N)$, we have $\hat{\mu} \in \calP_p(\bbR^N)$.
\item \label{item:isometry_GDS_FT}
(Isometry). The GDS-FT is an isometry with respect to $W_p(\cdot,\cdot)$, i.e., for $\mu, \nu \in \calP_{p}(\bbR^N)$,
\begin{align*}
W_p\left((\bU\T_{G})_{*}\mu,(\bU\T_{G})_{*}\nu\right) = W_p(\mu, \nu).
\end{align*}
\item \label{item.invertibility_GDS_FT} 
(Invertibility). The \gls{GDS}-FT is invertible, with inverse given by $\mu = (\bU_{G})_{*}
\widehat{\mu}$.
\item \label{item:composition}
(Associativity). For $\bF_{G}$ and $\bH_{G}$, we have
\begin{align*}
(\bF_{G})_{*}\left((\bH_{G})_{*}\mu\right) 
= (\bF_{G}\bH_{G})_{*}\mu.
\end{align*}
\item \label{item.lipschitz_stability}
(Lipschitz stability under linear pushforward). For any linear map $\bA\in\Real^{N\times N}$, the corresponding pushforward operator is Lipschitz continuous with respect to $W_p(\cdot,\cdot)$, i.e., for $\mu, \nu \in \calP_{p}(\bbR^N)$,
\begin{align}
\label{eq:linear_pushforward_lipschitz} 
W_p\!\left((\bA)_{\ast}\mu,(\bA)_{\ast}\nu\right) \leq \|\bA\|_{\mathrm{op}} W_p(\mu,\nu).
\end{align}
\end{enumerate}
\end{Proposition}
\begin{IEEEproof}
See \cref{sec.propo_proof} in the Appendix.
\end{IEEEproof}
% \red{[Move all proofs to the appendix or supplementary (if it is mainly algebraic).]}
  
Property \ref{item:isometry_GDS_FT} states that the GDS-FT preserves the Wasserstein distance between distributions, serving as the distributional counterpart of the Parseval identity.

\subsection{\gls{GDS} Convolutional Filter}

In classical \gls{GSP}, a graph convolutional filter applies a linear transformation that selectively amplifies or suppresses certain frequency components of a graph signal. We now define the \gls{GDS} version as follows.

\begin{Definition}[\gls{GDS} convolutional filter]
\label{def.GDS_filters}
Let $\bF_{G}: \mathbb{R}^N \to \mathbb{R}^N$ be a traditional graph convolutional filter. %, i.e., a polynomial in $\bS_{\calG}$.
The \emph{\gls{GDS} convolutional filter} is the pushforward map
\begin{align}
\calF_{G}: \calP_p(\mathbb{R}^N) \to \calP_p(\mathbb{R}^N), 
\quad \mu \mapsto (\bF_G)_{*}\mu.
\end{align}
\noindent$\rhd\ \bx \sim \mu,\ \bF_G\bx \sim \calF_G(\mu)$
\end{Definition}

When $\mu = \delta_{\bx}$ is a Dirac measure at a single signal $\bx$, we have ${\calF_G}(\mu) = \delta_{{\bF_G} \bx}$, thereby recovering the classical graph convolution ${\bF_G} \bx$ as a special case. More generally, a graph convolutional filter 
\begin{align}\label{graph_convolutional_filter}
\bF_\calG = \bU_G h(\bLambda_G) \bU_G\T,
\end{align} 
where $h(\cdot)$ is a polynomial of degree at most $N-1$, induces a \gls{GDS} convolutional filter  
\begin{align*}
\calF_{G}(\mu) & = (\bF_G)_{*}\mu = (\bU_G h(\bLambda_G) \bU_G\T)_{*}\mu \\ 
&= (\bU_G)_{*} h(\bLambda_G)_{*} (\bU_G\T)_{*}\mu,
\end{align*}
where the last equality follows from the associativity property in \cref{thm:properties_fixed}\ref{item:composition}. A GDS convolution can thus be viewed as a composition of three pushforward maps: a GDS-FT transports the input distribution $\mu$ to the spectral domain, where the filter $h(\bLambda_G)$ reshapes the distribution of frequency components, and finally an inverse GDS-FT transports the filtered distribution back to the vertex domain. 

\begin{Example}
Suppose $\mu$ is a multivariate Gaussian distribution with mean $\bm{m}$ and covariance matrix $\bSigma$. Applying the GDS convolution induced by the linear map $\bF_G$, we obtain that $\calF_{G}(\mu)$ is also Gaussian, with mean $\bF_G\bm{m}$ and covariance $\bF_G\bSigma\bF_G\T$ due to the property of jointly Gaussian distributions. Hence, $\calF_{G}$ modifies not only the central tendency of the signal, represented by its mean, but also its variability and dependencies, as captured by the covariance.
\end{Example}

\begin{Proposition}[Lipschitz continuity]\label{thm:lipschitz_GDS_filtering} 
Any graph convolutional filter $\bF_{G}$ is Lipschitz continuous with respect to the Wasserstein distance, i.e., 
\begin{align}
W_{p}\left((\bF_{G})_{*}\mu,(\bF_{G})_{*}\nu\right) \leq \|\bF_{G}\|_{\mathrm{op}}W_{p}(\mu,\nu)
\end{align}
for all $\mu, \nu \in \calP_p(\bbR^N)$.
\end{Proposition}
% \begin{IEEEproof}
% Applying \cref{eq:linear_pushforward_lipschitz} with $\bA=\bF_G$ directly gives
% \begin{align} 
% W_p\!\left((\bF_G)_{\ast}\mu,(\bF_G)_{\ast}\nu\right) \leq \|\bF_G\|_2 W_p(\mu,\nu). 
% \end{align}
% In particular, since $\bF_{G}$ is a polynomial in the symmetric matrix $\bS_{G}$, it shares the same orthonormal eigenbasis as $\bS_G$ and can be written as $\bF_G = \bU_G h(\bLambda_G)\bU_G\T$, where $h(\cdot)$ denotes the corresponding polynomial frequency response. Hence, $\|\bF_G\|_2 = \max_i |h(\lambda_i)|$, so the Lipschitz constant is determined by the maximum magnitude of the graph filter frequency response.
% \end{IEEEproof}
% 

\Cref{thm:lipschitz_GDS_filtering} follows immediately from \cref{eq:linear_pushforward_lipschitz}.
From \cref{graph_convolutional_filter}, $\|\bF_G\|_{\mathrm{op}} = \max_i |h(\lambda_i)|$, so the Lipschitz constant is determined by the maximum magnitude of the graph filter frequency response. \Cref{thm:lipschitz_GDS_filtering} guarantees that the GDS convolutional filter is Lipschitz continuous with respect to the Wasserstein distance. This property is particularly relevant in practice, where the true signal distribution $\mu$ is rarely known exactly and must be approximated, e.g., by empirical distributions from finite samples or parametric fits such as Gaussian or Gaussian mixture models. 
%Specifically, if such an approximation $\widetilde{\mu}$ satisfies $W_{p}(\mu, \widetilde{\mu})\leq \epsilon$, then the filtered output satisfies $W_{p}((\bF_{\calG})_{\#}\widetilde{\mu},(\bF_{\calG})_{\#}\mu)\leq \|\bF_{\calG}\|_{2} \epsilon$ thereby 
This ensures that distributional approximations at the input lead to reliable approximations at the output.

\subsection{Example Application: Graph Filter Learning}\label{sec.G_learning}
% Denote by $\mu_{\btheta} \in \mathcal P_{p}(\mathbb R^{N})$ the joint distribution over all graph nodes, parameterized by $\btheta$ and estimated from global observations. In the \gls{GDS} framework, graph filtering corresponds to the pushforward of a distribution through a graph filter: $\calF:\Real^{N}\times\calG_{N}\to \Real^{N}$, given by $\calF(\bx,\calG)=\bF_{\calG}\bx$,  yielding the filtered distribution
% $\mu_{\btheta}':= (\calF)_{\#}\calJ_{\mu_{\btheta},\mu_{\calG}}$. Our objective is to identify the graph filter $\bF_{\calG}$ such that the filtered distribution $\mu_{\btheta}'$ matches a target distribution $\mu^{\star}$ in the Wasserstein space, i.e., $\min_{\bF_{\calG}} W_{2}((\calF)_{\#}\calJ_{\mu_{\btheta},\mu_{\calG}},\mu^{\star})$. 

We present an illustrative graph filter learning example \cite{ortega_2022,Ort18} to demonstrate the GDS processing framework. Let $\mu_{\btheta} \in \mathcal P_{p}(\mathbb R^{N})$ denote the joint distribution over all graph nodes, parameterized by a weight vector $\btheta$. In the \gls{GDS} framework, filtering is formulated as the pushforward of this distribution through a graph filter $\bF_{G}:\mathbb R^{N}\to\mathbb R^{N}$, yielding $\mu_{\btheta}':= (\bF_G)_{*}\mu_{\btheta}$. Our objective is to identify $\bF_{G}$ such that the filtered distribution $\mu_{\btheta}'$ matches a target distribution $\mu^{\star}$ in Wasserstein space:
\begin{align}\label{eq.filter_formulation}
\min_{\bF_{G}} &\ W_{2}((\bF_G)_{*}\mu_{\btheta},\mu^{\star}).
\end{align}
This problem formulates a graph filter learning for GDSs, where the learned filter transports the input distribution $\mu_{\btheta}$ to align with the target distribution $\mu^{\star}$. In this sense, $\bF_{G}$ is learned at the distribution level, enabling the filter to capture graph-dependent transformations of uncertainty and correlations across nodes, rather than only manipulating individual signal realizations.

To make the optimization problem in \cref{eq.filter_formulation} tractable, we introduce explicit parametric forms for the joint distributions $\mu_{\btheta}$ and $\mu^{\star}$. Specifically, we consider two representative schemes: (1) copula-based GDS filter learning (\cref{subsec.copula}), in which the joint distribution is constructed by coupling prescribed marginal distributions through a copula, thereby modeling inter-node dependence separately from the marginals; and (2) \gls{GMM}-based GDS filter learning (\cref{subsec.GMMs}), in which the joint distribution is represented as a mixture of Gaussian components to provide high expressive flexibility.

\subsubsection{Copula-based GDS filter learning}\label{subsec.copula}

Assume that each node $v_i \in V$, $i=1,\dots,N$, is associated with a marginal distribution $\mu_i(\cdot;\btheta_i)\in\calP_{p}(\bbR)$, parameterized by $\btheta_i$ and estimated from local observations. Because synchronous measurements across all nodes may be unavailable, direct estimation of the joint distribution is generally infeasible. We therefore construct a joint distribution $\mu_{\btheta,\kappa}\in \calP_{p}(\bbR^{N})$ by combining the marginals with a copula density $c_\kappa$, which captures inter-node dependence:
\begin{align*}
\mu_{\btheta,\kappa}(x)=c_{\kappa}\left(F_{1}(x_{1};\btheta_{1}),\dots,F_{N}(x_{N};\btheta_{N})\right) \prod_{i=1}^{N} \mu_i(x_i;\btheta_i)
\end{align*}
where $F_i(\cdot;\btheta_i)$ denotes the \gls{CDF} of $\mu_i(\cdot;\btheta_i)$, $\btheta=(\btheta_i)_{i=1}^N$, and $\kappa$ parametrizes the copula. By Sklar's theorem \cite{Sklar1959,Durante2013}, the marginals of $\mu_{\btheta,\kappa}$ are exactly $\mu_i(\cdot;\btheta_i)$ for $i=1,\dots,N$.

In our \gls{GDS} framework, the filtered GDS is defined as $\mu_{\btheta,\kappa}':= (\bF_G)_{*}\mu_{\btheta,\kappa}$. We then reformulate the objective to jointly estimate the graph filter $\bF_{G}$ and the copula density $c_{\kappa}$ such that the filtered GDS $\mu_{\btheta,\kappa}'$ matches a target distribution $\mu^{\star}$:
\begin{align}
\begin{aligned}\label{eq.pro_formu}
\min_{\bF_{G},\kappa} &\ W_{2}((\bF_G)_{*}\mu_{\btheta,\kappa},\mu^{\star})\\
\ST &\ (\pi_i)_{*} \mu_{\btheta,\kappa} = \mu_i, \quad i=1,\dots,N,
\end{aligned}
\end{align}
where $\pi_i$ denotes the projection onto the $i$-th coordinate. This formulation characterizes graph filter learning in which the marginals encode local uncertainty, the copula captures inter-node dependencies, and the learned filter aligns the resulting GDS with the target distribution.

To obtain an explicit form of \cref{eq.pro_formu}, we adopt a Gaussian parameterization. Each node $v_i$ is associated with a marginal distribution $\mu_i=\calN(m_i,\sigma_i^2)$, where $\btheta_i= (m_i, \sigma_i^2)$ are estimated from local observations at node $v_i$. Inter-node dependence is modeled by a Gaussian copula with correlation matrix $\bm{R}$, which induces the joint distribution $\mu=\calN(\bm{m},\bm{\Sigma})$, where $\btheta = (\bm{m}, \bSigma)$, $\bm{m}=(m_1,\dots,m_N)\T$ and $\bm{\Sigma}=\bm{D}\bm{R}\bm{D}$ with $\bm{D}=\mathrm{diag}(\sigma_1,\dots,\sigma_N)$. Applying the graph filter $\bF_G$ to $\mu$ produces another Gaussian distribution:
\begin{align*}
(\bF_G)_{*}\mu =\calN(\bF_{G}\bm{m}, \bF_{G}\bm{\Sigma}\bF\T_{G} ).
\end{align*}

When the target distribution is also modeled as Gaussian, i.e., $\mu^{\star}=\calN(\bm{m}^{\star},\bm{\Sigma}^{\star})$, the squared Wasserstein distance admits a closed-form expression (see \cref{eg:imd}). Consequently, \cref{eq.pro_formu} reduces to:
\begin{align}
\begin{aligned}
\label{eq.opt}
\min_{\bF_{G}, \bm{R}} &\ \norm{\bF_{G}\bm{m}-\bm{m}^{\star}}_{2}^{2} + \Tr\Big\{\bF_{G}\bm{\Sigma}\bF\T_{G} +\bm{\Sigma}^{\star}\\
&\qquad\qquad- 2 \big( (\bm{\Sigma}^\star)^{1/2}\bF_{G}\bm{\Sigma}\bF_{G}\T(\bm{\Sigma}^\star)^{1/2} \big)^{1/2}\Big\}\\
\ST &\ \bm{R}=\bm{R}\T,\quad \bm{R}\succeq \mathbf{0}, \quad \diag(\bm{R})=\bone.
\end{aligned}
\end{align}

We solve this optimization problem using alternating minimization over $\bF_{G}$ and $\bm{R}$, with each subproblem updated via gradient descent \cite{Lecun1998}. For a graph convolutional filter $\bF_G$ of the form specified in \cref{graph_convolutional_filter}, the optimization is performed over the coefficients of the polynomial $h(\cdot)$.
The complete procedure is summarized in \cref{alg:Cop_filter_learning}.

\begin{algorithm}[tb]
\caption{Copula-based GDS graph filter learning algorithm}
\label{alg:Cop_filter_learning}
\begin{algorithmic}[1]
\State Input: Target mean $\bm{m}^\star$ and covariance $\bm{\Sigma}^\star$, marginal means $\bm{m}$, variances $\bm{D} = \mathrm{diag}(\sigma_1, \dots, \sigma_N)$, learning rates $\eta_1$ and $\eta_2$, convergence tolerance $\epsilon$, positive threshold $\delta$. $\calL$ is the objective function in \cref{eq.opt}. 
\State Output: Optimized $\widetilde{\bF}_{G}$
\State Initialize $\bF_{G}^{(0)}$, $\bm{R}^{(0)}$, $u = 0$
\Repeat
\State $\bm{\Sigma}^{(u)} \gets \bm{D} \bm{R}^{(u)} \bm{D}$
\State $\calL^{(u)} \gets \cref{eq.opt}$ 
% using
% \begin{align*}
%     \calL^{(u)} = \norm{\bF_{\calG}^{(u)}\bm{m} - \bm{m}^{\star}}_2^2 
%     + \Tr\Big( \bF_{\calG}^{(u)} \bm{\Sigma}^{(u)} \bF_{\calG}^{(u)\top} + \bm{\Sigma}^\star\\ 
%     - 2\big((\bm{\Sigma}^\star)^{1/2} \bF_{\calG}^{(u)} \bm{\Sigma}^{(u)} \bF_{\calG}^{(u)\top} (\bm{\Sigma}^\star)^{1/2} \big)^{1/2} \Big)
% \end{align*}
\State $\bF_{G}^{(u+1)} \gets \bF_{G}^{(u)} - \eta_1 \nabla_{\bF_{G}} \calL^{(u)}$
\State $\bm{R}^{(u+1)} \gets \bm{R}^{(u)} - \eta_2 \nabla_{\bm{R}} \calL^{(u)}$
\State Project $\bm{R}^{(u+1)}$ to a valid correlation matrix:
% \begin{itemize}
%     \item Symmetrize: $\bm{R}^{(u+1)} \gets (\bm{R}^{(u+1)} + \bm{R}^{(u+1)\top}) / 2$
%     \item Project to PSD: $\bm{R}^{(u+1)} \gets \bm{V} \cdot \operatorname{diag}(\max(\bm{\lambda}, \epsilon)) \cdot \bm{V}\T$
%     \item Normalize: $\bm{R}^{(u+1)} \gets \bm{S}^{-1} \bm{R}^{(u+1)} \bm{S}^{-1}$,\\
%     $\bm{S} = \operatorname{diag}\left(\sqrt{\operatorname{diag}(\bm{R}^{(u+1)})}\right)$
%     \item Set diagonal: $\diag(\bm{R}^{(u+1)}) \gets \mathbf{1}$
% \end{itemize}
\State $\bm{R}^{(u+1)} \gets \tfrac{1}{2} \left( \bm{R}^{(u+1)} + \bm{R}^{(u+1){\T}} \right)$ \hfill \Comment{Symmetrize}
\State $\bm{R}^{(u+1)} \gets \bm{V} \diag(\max(\bm{\lambda}, \delta)) \bm{V}\T$ \hfill \Comment{PSD projection}
\State $\bm{S} \gets \diag\left( \sqrt{ \diag(\bm{R}^{(u+1)}) } \right)$
\State $\bm{R}^{(u+1)} \gets \bm{S}^{-1} \bm{R}^{(u+1)} \bm{S}^{-1}$ \hfill \Comment{Normalize}
\State $\diag(\bm{R}^{(u+1)}) \gets \mathbf{1}$ \Comment{Set diagonal}
\State $u \gets u + 1$
\Until{$\norm{\calL^{(u)} - \calL^{(u-1)}} \le \epsilon$}
\end{algorithmic}
\end{algorithm}

\subsubsection{GMM-based GDS filter learning}\label{subsec.GMMs}

To balance expressive power and analytical tractability, we propose to model the joint distribution across all graph nodes using a \gls{GMM}. Specifically, we assume 
\begin{align}
\label{eq.mu_gmm}
\mu_{\btheta} = \sum_{k=1}^{K}a_{k}\xi_{k} \in \mathrm{GMM}_{K}(\Real^{N}),    
\end{align}
where each component $\xi_{k}$ is a multivariate Gaussian $\xi_{k}=\calN(\bm{m}_{k},\bSigma_{k})$, and the mixture weights satisfy $a_{k}\geq 0$ and $\sum_{i=1}^{K}a_{k}=1$. 

With this model, the marginal distribution of the $i$-th node $\mu_{i}:=(\pi_{i})_{*}\mu_{\btheta}$ is a univariate GMM with the same mixture weights, i.e.,
\begin{align}
\begin{aligned}
\label{eq.mu_igmm}
\mu_{i}=\sum_{k=1}^{K} a_{k}&\calN(m_{ik},\sigma_{ik}^{2}),\quad i=1,\ldots, N, 
\end{aligned}   
\end{align}
where $m_{ik}$ and $\sigma_{ik}^{2}$ denote the mean and variance of the $i$-th coordinate under the $k$-th Gaussian component $\xi_k$. Moreover, within each component $k$, the dependence structure is characterized by a Gaussian copula with correlation matrix $\bR_{k}$. Equivalently, the covariance can be decomposed as $\bSigma_{k}=\bD_{k}\bR_{k}\bD_{k}$ with $\bD_{k}:=\diag(\sigma_{1k},\cdots,\sigma_{nk})$, and the component mean vector is $\bm{m}_{k}=(m_{1k},\cdots,m_{Nk})\T$. 

Similarly, the target distribution is modeled as $\mu^{\star}=\sum_{l=1}^{L}b_{l}\xi_{l}^{\star}\in \mathrm{GMM}_{L}(\Real^{N})$ with at most $L$ Gaussian components, where each component is given by $\xi_{l}^{\star}=\calN(\bm{m}_{l}^{\star},\bSigma_{l}^{\star})$.  
Under these assumptions, the filter learning problem \cref{eq.filter_formulation} is
\begin{align}
\begin{aligned}
\label{eq.GF_GMM}
&\qquad\qquad \min_{\bF_{G}} \ W_{2}((\bF_G)_{*}\mu_{\btheta},\mu^{\star})\\
&\ST \mu_{\btheta}\in \mathrm{GMM}_{K}(\Real^{N})~\text{with}~(\pi_i)_{*} \mu_{\btheta}=\mu^{(i)}\\ 
&\qquad\qquad ~\text{and}~\mu^{\star}\in \mathrm{GMM}_{L}(\Real^{N}),
\end{aligned}
\end{align}
where $(\bF_G)_{*}\mu_{\btheta}=\sum_{k=1}^{K} a_{k} (\bF_G)_{*}\xi_{k}$.

For tractability, we approximate the Wasserstein distance $W_{2}$ by the Mixture-Wasserstein distance $MW_{2}$ defined in \cref{eq.MW2}. It admits efficient computation via component-wise optimal transport and provides a practical surrogate for $W_{2}$ \cite{Vil09,Cuturi2013,Chen2018}. This leads to the following relaxed formulation:
% \begin{align}
% \begin{aligned}
% \label{for.WM_2}
%     &\min_{\bF_{\calG},\omega \in \Pi(\ba,\bb)} \ \sum_{k,l}\omega_{kl} W_{2}((\bF_{\calG})_{\#}\mu_{k},\nu_{l}^{\star})\\
%     &\ST \ba=\left(a_1,\dots,a_{K}\right)\T, \bb=\left(b_1,\dots,b_{L}\right)\T \\
%     & \quad \Pi(\ba,\bb)=\set{\omega\in\Real_{+}^{K\times L}\given \omega\bone = \ba~\text{and}~\omega\T\bone=\bb}.
% \end{aligned}
% \end{align}

\begin{align}
\begin{aligned}
\label{for.WM_2}
&\min_{\bF_{G},\{\bR_{k}\}_{k=1}^{K},\bomega \in \Pi(\ba,\bb)} \ \sum_{k,l}\bomega_{kl} W_{2}^{2}((\bF_{G})_{*}\xi_{k},{\xi_{l}}^{\star})\\
&\ST \ba=\left(a_1,\dots,a_{K}\right)\T, \bb=\left(b_1,\dots,b_{L}\right)\T \\
&\quad \Pi(\ba,\bb)=\set{\bomega\in\Real_{+}^{K\times L}\given \bomega\bone = \ba~\text{and}~\bomega\T\bone=\bb}\\
&\quad \bm{\bR}_{k}=\bm{\bR}_{k}\T,\quad \bm{\bR}_{k}\succeq \mathbf{0}, \quad \diag(\bm{\bR}_{k})=\bone, \forall k.
\end{aligned}
\end{align}
Here, $W_{2}((\bF_{G})_{*}\xi_{k},\xi_{l}^{\star})^2$ admits the close-form expression:
\begin{align}
\begin{aligned}
\label{eq.cost_matrix}
&W_{2}^{2}((\bF_{G})_{*}\xi_{k},\xi_{l}^{\star})= \norm{\bF_{G}\bm{m}_{k}-\bm{m}_{l}^{\star}}_{2}^{2} \\
& + \Tr\braces*{ \bF_{G}\bm{\Sigma}_{k}\bF\T_{G} +\bm{\Sigma}_{l}^{\star} - 2\big( (\bm{\Sigma}_{l}^{\star})^{1/2}\bF_{G}\bm{\Sigma}_{k}\bF_{G}\T(\bm{\Sigma}_{l}^{\star})^{1/2} \big)^{1/2} }.   
\end{aligned}
\end{align}

We solve this problem via alternating minimization. Specifically, we update the graph filter $\bF_{G}$ and the correlation matrices $\{\bR_{k}\}_{k=1}^{K}$ using gradient descent, while recomputing the transport plan $\bomega$ at each iteration by solving the entropically regularized optimal transport problem \cite{Vil09,Cuturi2013}, i.e.,
\begin{align}
\label{eq.op_plan}
\bomega^{\star}\in \argmin_{\bm{\omega}\in\Pi(\ba,\bm{b})}\;
\ip*{\bm{\omega}}{\bC} + \varepsilon \sum_{k,l}\bomega_{kl}(\log\bomega_{kl}-1),
\end{align}
where $\bC$ is the cost matrix (i.e., $\bC_{kl}=W_{2}^{2}((\bF_{G})_{*}\xi_{k},\xi_{l}^{\star})$) defined in \cref{eq.cost_matrix}, $\varepsilon>0$ is the entropic regularization parameter, $\Pi(\ba,\bm{b})$ denotes the set of couplings with prescribed marginals $\ba$ and $\bm{b}$, and $\ip{\cdot}{\cdot}$ denotes the Frobenius inner product. We compute $\bomega^{\star}$ efficiently using the Sinkhorn algorithm \cite{Sinkhorn1974,Altschuler2017}. The overall procedure is summarized in \cref{alg:gmms_filter_learning}.

%and is referred to as the \emph{GMMs-based GDS (GDS-GMMs)} graph filter learning algorithm, where $\calL$ denotes the objective function in \cref{for.WM_2}. %The learned graph filter $\widetilde{\bF}_{\calG}$ can be readily applied to downstream prediction tasks, as shown in \cref{eq.for_predict}.

\begin{algorithm}[tb]
\caption{GMMs-based GDS graph filter learning algorithm}
\label{alg:gmms_filter_learning}
\begin{algorithmic}[1]
\State Input: Source GMM $\{(\bm{m}_{k},\bm{\Sigma}_{k},a_{k})\}_{k=1}^{K}$, target GMM $\{(\bm{m}_{l}^\star, \bm{\Sigma}_{l}^\star, b_{l})\}_{l=1}^{L}$, learning rate $\eta_{1}$ and $\eta_{2}$, convergence tolerance $\epsilon$, regularization coefficient $\varepsilon$, positive threshold $\delta$. $\calL$ is the objective function in \cref{for.WM_2}.
% , the number of Gaussian components in GMMs $L$ and $K$
\State Output: Optimized $\widetilde{\bF}_{G}$
\State Initialize $\bF_{G}^{(0)}$, $(\bm{R}^{(0)}_{k})_{k=1}^{K}$ and $u = 0$
\Repeat
\State $W_{2}^{2}((\bF_{G}^{(u)})_{\ast}\xi_{k},\xi_{l}^{\star})^{2} \gets \cref{eq.cost_matrix}$ \hfill \Comment{Compute cost matrix}
% \State $\omega^{(u)} \gets \text{solve \cref{for.WM_2} via Sinkhorn Algorithm}$ \hfill \Comment{Update transport plan}
% \State Solve \cref{for.WM_2} via entropic OT with Sinkhorn Algorithm:
\State $\bomega^{(u)} \gets \cref{eq.op_plan}$    \hfill \Comment{Update transport plan} 
\State $\calL^{(u)} \gets \text{compute}~\cref{for.WM_2}$ with $\bomega^{(u)}$  \hfill \Comment{Update loss}
\State $\bF_{G}^{(u+1)} \gets \bF_{G}^{(u)} - \eta_{1} \nabla_{\bF_{G}} \calL^{(u)}$
\For {$k=1,\dots,K$}
\State $\bm{R}^{(u+1)}_{k} \gets \bm{R}^{(u)}_{k} - \eta_{2} \nabla_{\bm{R}_{k}} \calL^{(u)}$
\State Project $\bm{R}_{k}^{(u+1)}$ to a valid correlation matrix:
\State $\bm{R}_{k}^{(u+1)} \gets \tfrac{1}{2} \left( \bm{R}_{k}^{(u+1)} + \bm{R}_{k}^{(u+1){\T}} \right)$ \hfill \Comment{Symmetrize}
\State $\bm{R}^{(u+1)}_{k} \gets \bm{V}_{k} \diag(\max(\bm{\lambda}_{k}, \delta)) \bm{V}\T_{k}$ \hfill \Comment{PSD}
\State $\bm{S}_{k} \gets \diag\left( \sqrt{ \diag(\bm{R}^{(u+1)}_{k}) } \right)$
\State $\bm{R}^{(u+1)}_{k} \gets \bm{S}_{k}^{-1} \bm{R}^{(u+1)}_{k} \bm{S}^{-1}_{k}$ \hfill \Comment{Normalize}
\State $\diag(\bm{R}^{(u+1)}_{k}) \gets \mathbf{1}$ \Comment{Set diagonal}
\EndFor
\State $u \gets u + 1$
\Until{$\norm{\calL^{(u)} - \calL^{(u-1)}} \le \epsilon$}
\end{algorithmic}
\end{algorithm}

\section{The Generalized GDS Framework}
\label{sec.general_gds}

We next generalize the GDS processing framework introduced in \cref{sec:gds_framework} by allowing the graph topology to follow a signal-dependent distribution \cite{JiTayOrt23}.

\subsection{Signal Adaptive Graph Structures}

We propose to treat the graph as a data-dependent representation whose connectivity varies with the observed signal (see \cref{Fig.SAGS}). In practice, statistical dependencies among nodes are often induced by latent physical or functional mechanisms that manifest in the measurements themselves. Consequently, a graph constructed from the data---e.g., by linking nodes with strong correlation or coherence---may vary across signal realizations, even when the vertex set and sensing locations remain fixed. This motivates modeling the graph as a \emph{signal-adaptive} (and potentially random) object:
given a graph signal $\bx\in\Real^{N}$, we associate with it a conditional distribution $\nu_{\bx}$ over graphs that captures the plausible connectivity patterns consistent with $\bx$. We formalize this notion via \emph{\glspl{SAGS}}. Throughout this paper, we view $\calG_N$ as parametrized by adjacency matrices, and identified with a Euclidean space.

\begin{Definition}[Signal adaptive graph structure (SAGS)]
\label{def.SAGS}
Let $V=\set{v_{1},\dots,v_{N}}$ be an ordered set of $N$ vertices, and let $\calG_N$ denote the space of all weighted graphs on the fixed vertex set $V$. A SAGS $\calA$ is, for each $\bx \in \Real^N$, a conditional probability distribution $\nu_{\bx}(\cdot) = \calA(\bx, \cdot) \in \calP_{p}(\calG_{N})$.
$\rhd\ G \mid \bx \sim \nu_{\bx}$
\end{Definition}

More precisely, adopting probability-theoretic notations \cite[Section 4.1.3]{Dur19}, for any Borel set $B \subseteq \calG_N$, $\calA(\bx, B)$ is a version of $\P(G \in B \mid \sigma(\bx))$, where $\sigma(\bx)$ denotes the $\sigma$-algebra generated by the random element $\bx$. Moreover, for \emph{each} $\bx$, $\nu_{\bx}(\cdot) := \calA(\bx, \cdot)$ defines a probability measure on $\calG_N$. 
Note that in this work, for simplicity, we assume that $\nu_{\bx}$ is well-defined for every $\bx\in\Real^N$, instead of only for $\mu$-a.e. $\bx$. This simplification avoids dependency of $\calA$ on $\mu$, and measure-theoretic difficulties, which require additional technical assumptions to resolve. 
In the sequel, we adopt the simplified notation introduced in \cref{def.SAGS}.

If $\nu_{\bx}=\nu$ for all $\bx$ (i.e., it is independent of $\bx$), we call it a constant SAGS \cite{JiTayOrt23}. Moreover, if $\nu_{\bx}=\delta_{G}$ for a fixed graph $G\in\calG_{N}$, then the SAGS framework reduces to the classical GSP setting with a single deterministic underlying graph.

For any $\mu\in \calP_{p}(\Real^{N})$, we have an associated joint distribution $\calA^*\mu$ on $\Real^{N} \times \calG_{N}$, given by
\begin{align}
\label{eq:cmi}
(\calA^*\mu)(B) = \int_{B} \ud\nu_{\bx}(G) \ud\mu(\bx)
\end{align}
for any Borel set $B \subseteq \Real^{N} \times \calG_{N}$. This is the pullback map of distributions $\calP_{p}(\mathbb{R}^N)\to \calP_{p}(\mathbb{R}^N\times \calG_{N})$ (i.e., the joint distribution of $(\bx, G)$ in the $p$-th order Wasserstein space), as confirmed by the following result.

\begin{Lemma}\label{Lemma_pullback}
For any GDS $\mu \in \calP_p(\Real^N)$ and SAGS $\calA$ such that 
\begin{align}\label{eq:sags-moment}
 \int_{\Real^{N}} W_p^p(\nu_{\bx}, \delta_{G_0})\ud\mu(\bx) < \infty   
\end{align}
for some fixed $G_0 \in \calG_N$, the distribution $\calA^*\mu$ defined in \cref{eq:cmi} belongs to $\calP_{p}(\Real^{N}\times \calG_{N})$. $\rhd\ (\bx, G) \sim \calA^*\mu$
\end{Lemma}

\begin{IEEEproof}
    See \cref{proof_lemma1} in the Appendix.
\end{IEEEproof}

% \begin{Example} \label{eg:inn}

% \item \red{\label{it:agc} A SAGS $\calA_{\calG}=(\mu_{\calG}(\cdot|\bx))_{\bx\in\mathbb{R}^N}$ is \emph{locally constant} if for almost every $\bx \in \mathbb{R}^N$, there is an open neighborhood $U_{\bx}$ of $\bx$ such that for every $\by\in U_{\bx}$, we have $\mu_{\calG}(\cdot|\by)=\mu_{\calG}(\cdot|\bx)$. Intuitively, for such an $\calA_{\calG}$, the signal space $\mathbb{R}^N$ can be (almost) partitioned into open subsets on each of which $\calA_{\calG}$ is a constant.}
% \end{Example}

An SAGS $\calA$ encodes the structural information of the underlying graph in a signal-dependent manner. For a given signal $\bx$, it determines the plausible graph structures on $V$ through the measure $\nu_{\bx}$. This concept is illustrated in \cref{Fig.SAGS}. The next subsection describes how \gls{GDS} transformations are performed under this signal-adaptive graph structure, extending the framework of \cref{sec:gds_framework}.

\begin{figure}[!htb]
\centering

\begin{tikzpicture}[scale=0.65]

% background
\fill[gray!15] (-0.2,-0.2) rectangle (12,8);

% dashed separators
% \draw[dash dot dot, thick] (6,0) -- (6,8);  % 
% \draw[dash dot dot, thick] (0,3.5) -- (12,3.5);

% ---------------- Top-left: x1 ----------------
% blue stems (x,y,length)
\foreach \x/\y/\l in {
1.5/6.3/1.4,
2.5/6.8/0.95,
3.5/6.5/0.5,
4.5/5.8/0.3
}
{
\draw[blue, line width=0.8pt] (\x,\y) -- ++(0,\l);
\fill[red] (\x,\y) circle (0.10);
}
% black stems
\foreach \x/\y/\l in {
2/5.6/1.2,
2.8/5.1/0.7,
3.8/5.1/0.4
}
{
\draw[blue, line width=0.8pt] (\x,\y-\l) -- (\x,\y);
\fill[red] (\x,\y) circle (0.10);
}

% \node[font=\bfseries\large] at (5,4.35) {$\bx_1$};

\node[font=\fontsize{9}{10}\selectfont] at (3,4) {(a) $\bx_1$};
\node[font=\fontsize{9}{10}\selectfont] at (3+6,4) {(b) Graph samples from $\nu_{\bx_1}$};

% -------------Top right --------------------------------------------------------
% sample 1
\coordinate (A) at (1.5+6,6.3+1);
\coordinate (B) at (2.5+6,6.8+1);
\coordinate (C) at (3.5+6,6.5+1);
\coordinate (D) at (4.5+6,5.8+1);
\coordinate (E) at (3.8+6,5.3+1+0.2);
\coordinate (F) at (2.8+6,5.3+1+0.2);
\coordinate (G) at (2+6, 5.8+1+0.2);
\draw[gray, line width=0.9pt]
(A)--(B)--(C)--(D)--(E)--(F)--(G);
\foreach \P in {A,B,C,D,E,F,G}
\fill[red] (\P) circle (0.1);
% sample 2
\coordinate (A) at (1.5+6,6.3-0.8);
\coordinate (B) at (2.5+6,6.8-0.8);
\coordinate (C) at (3.5+6,6.5-0.8);
\coordinate (D) at (4.5+6,5.8-0.8);
\coordinate (E) at (3.8+6,5.3-0.8+0.2);
\coordinate (F) at (2.8+6,5.3-0.8+0.2);
\coordinate (G) at (2+6,5.8-0.8+0.2);
\draw[gray, line width=0.9pt]
(A)--(B)--(C)--(D)--(E)--(F)--(G)
(C)--(E);
\foreach \P in {A,B,C,D,E,F,G}
\fill[red] (\P) circle (0.1);

% -------------------------Bottom left-----------------------------------------------

\foreach \x/\y/\l in {
1.5/6.3-4.4/1.4,
2.5/6.8-4.4/0.8,
3.5/6.5-4.4/0.5,
4.5/5.8-4.4/0.3
}
{
\draw[blue, line width=0.8pt] (\x,\y) -- ++(0,\l);
\fill[red] (\x,\y) circle (0.10);
}
% black stems
\foreach \x/\y/\l in {
2/5.6-4.4/1.2,
2.8/5.1-4.4/0.9,
3.8/5.1-4.4/0.5
}
{
\draw[blue, line width=0.8pt] (\x,\y) -- ++(0,\l);
\fill[red] (\x,\y) circle (0.10);
}

% ------------------------ Bottom right ------------------------------------------
% sample 1
\coordinate (A) at (1.5+6,6.3+1-4.2);
\coordinate (B) at (2.5+6,6.8+1-4.2);
\coordinate (C) at (3.5+6,6.5+1-4.2);
\coordinate (D) at (4.5+6,5.8+1-4.2);
\coordinate (E) at (3.8+6,5.3+1-4.2+0.2);
\coordinate (F) at (2.8+6,5.3+1-4.2+0.2);
\coordinate (G) at (2+6, 5.8+1-4.2+0.2);
\draw[gray, line width=0.9pt]
(A)--(B)--(C)--(D)--(E)--(F)--(G)
(A)--(G);
\foreach \P in {A,B,C,D,E,F,G}
\fill[red] (\P) circle (0.1);
% sample 2
\coordinate (A) at (1.5+6,6.3-0.8-4.15);
\coordinate (B) at (2.5+6,6.8-0.8-4.15);
\coordinate (C) at (3.5+6,6.5-0.8-4.15);
\coordinate (D) at (4.5+6,5.8-0.8-4.15);
\coordinate (E) at (3.8+6,5.3-0.8-4.15+0.2);
\coordinate (F) at (2.8+6,5.3-0.8-4.15+0.2);
\coordinate (G) at (2+6,5.8-0.8-4.15+0.2);
\draw[gray, line width=0.9pt]
(A)--(B)--(C)--(D)--(E)--(F)--(G)
(C)--(E) (A)--(G) (B)--(F);
\foreach \P in {A,B,C,D,E,F,G}
\fill[red] (\P) circle (0.1);

\node[font=\fontsize{9}{10}\selectfont] at (3,0) {(c) $\bx_2$};
\node[font=\fontsize{9}{10}\selectfont] at (3+6,0) {(d) Graph samples from $\nu_{\bx_2}$};

\end{tikzpicture}
\caption{An illustration of SAGS. For two distinct signals $\bx_1, \bx_2$, the likelihood of a graph structure is determined by the smoothness of the signal on that structure, resulting in distinct distributions $\nu_{\bx_1}, \nu_{\bx_2}$. (a)(c) Two distinct signals. (b)(d) Representative graph structures sampled from $\nu_{\bx_{1}}$ and $\nu_{\bx_{2}}$, respectively.}
\label{Fig.SAGS}
\end{figure}
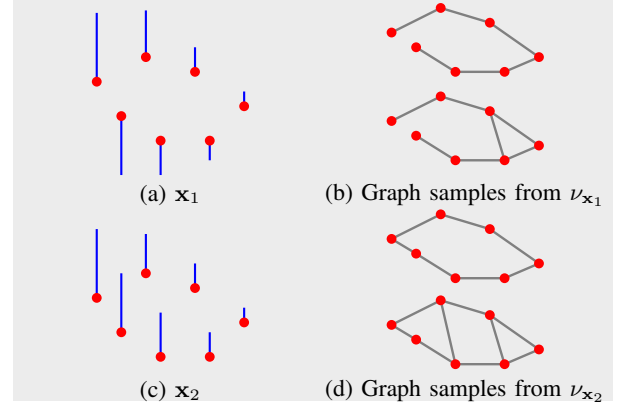

\subsection{Generalized GDS Processing Framework}

\subsubsection{Generalized GDS transforms} We extend the \gls{GDS} Fourier transform and \gls{GDS} convolutional filtering to the SAGS setting. These operations, referred to as the generalized \gls{GDS} Fourier transform and generalized \gls{GDS} convolutional filtering, are defined below.

\begin{Definition}[Generalized \gls{GDS} Fourier transform]
\label{Def.GGDS_FT}
For a given graph signal $\bx\in\Real^N$ and graph $G\in\calG_N$, let $\Phi(\bx, G):=\bU_{G}\T \bx$ denote the GFT of $\bx$ with respect to $G$. For a GDS $\mu \in \calP_p(\bbR^N)$ and a SAGS $\calA$, the generalized \emph{\gls{GDS-FT}} is the pushforward measure
\begin{align}\label{eq.GGDS-FT}
\hat{\mu}:= (\Phi_{*} \circ \calA^{*})\mu,
\end{align}
i.e., for any Borel set $ B \in \calB(\bbR^N)$, 
\begin{align}
\widehat{\mu}(B)
&=(\calA^{*}\mu) \left(\Phi^{-1}(B)\right) \nn
&=(\calA^*\mu)\parens*{\set*{(\bx, G) \given \bU_{G}\T \bx \in B}}.
\end{align}
\noindent$\rhd\ \bx \sim \mu,\ G\mid\bx \sim \nu_{\bx},\ \bU_G\T\bx \sim \hat{\mu}$
\end{Definition}

Intuitively, the generalized \gls{GDS-FT} transforms a distribution from the vertex domain to the frequency domain by jointly accounting for both the signal and the underlying graph structure through the joint distribution $\calA^{*}\mu$. For each pair $(\bx, G)$ drawn from the distribution $\calA^{*}\mu$, the signal $\bx$ is projected onto the eigenbasis $\bU_{G}$ of the corresponding graph, and the resulting distribution of these projections defines $\widehat{\mu}$.  
% The inverse \gls{GDS-FT} is given by $\mu = (\Phi^{-1})_{\#}\widehat{\mu}$.  

% \subsubsection{Generalized GDS convolutional filtering}
% We now extend the GDS convolutional filtering defined in \cref{def.GDS_filters} to the signal-adaptive setting as follows.

\begin{Definition}[Generalized \gls{GDS} convolutional filter]\label{def.GGDS_filters}
Let $\bF_G$ be a graph convolutional filter (i.e., a polynomial of $\bS_{G}$) for each graph $G\in\calG_N$ and let 
\begin{align}\label{eq.calJ}
\calJ : \Real^{N}\times\calG_{N}\to \Real^{N},\quad (\bx, G) \mapsto \bF_{G}\bx.
\end{align}
For a SAGS $\calA$, a \emph{generalized \gls{GDS} convolutional filter} is a mapping $\calT := \calT_{\calJ,\calA}$ defined by
\begin{align}\label{eq.calT}
\calT : \calP_p(\mathbb{R}^N) \to \calP_p(\mathbb{R}^N), \quad \mu \mapsto (\calJ_{*}\circ \calA^{\ast})\mu.
\end{align}
I.e., for any Borel set $B \in \calB(\bbR^N)$, the filter output yields
\begin{align*}
\calT(\mu)(B)
&=(\calA^{*}\mu) \left(\calJ^{-1}(B)\right) \\
&= (\calA^{*}\mu)\parens*{\set*{(\bx, G) \given \bF_G\bx \in B}}. 
\end{align*}
\noindent$\rhd\ \bx \sim \mu,\ G\mid\bx \sim \nu_{\bx},\ \calJ(\bx,G)=\bF_G\bx \sim \calT(\mu)$
\end{Definition}

Analogously to the Fourier transform, a generalized GDS convolutional filter maps an input GDS $\mu$ to the filtered GDS defined as the distribution of $\bY = \bF\bx$, where $(\bx, G)$ is drawn from the joint distribution $\calA^{\ast}\mu$. The filter thus propagates uncertainty in both the signal and the graph topology through the same local aggregation mechanism as classical graph filtering, applied pointwise to each realization of $(\bx, G)$. In particular, when the graph is deterministic, i.e., $\nu_\bx = \delta_{G}$, both \cref{def.GDS_FT,def.GDS_filters} are recovered as special cases.

\cref{tab:dic} presents a dictionary that maps core concepts of traditional GSP to their counterparts in the proposed generalized GDS framework. Using this dictionary, existing GSP models and formulations can be systematically translated into the generalized GDS setting.

\begin{table}[!htb]
\centering
\caption{Mapping between GSP and GDS} \label{tab:dic}
\centering
\begin{tabular}{rl} 
\toprule
Traditional GSP & GDS framework  \\ 
\midrule
Signal &\gls{GDS} (\cref{def.gds}): $\mu$ \\ 
Graph & SAGS (\cref{def.SAGS}): $\calA$  \\ 
GFT & Generalized GDS-FT (\cref{Def.GGDS_FT})   \\
Filter & Generalized GDS filter (\cref{def.GGDS_filters})\\
% \hline
% Convolution & GDS convolution (\cref{eq:gpm}): $\mathfrak{c}=(\calA,P\circ s)$ \\
%Total variation & See \cite{JiZhaLee25} \\
\bottomrule
\end{tabular}
\end{table}   

\subsection{Analytic Properties Of Generalized GDS Transforms}

In this subsection, we study the analytic properties of generalized GDS transforms, which provide the theoretical foundation for modeling GDSs with empirical or mixed Gaussian distributions. To establish the main results in a unified framework, we first introduce a general notion of GDS transformation under which both the generalized GDS Fourier transform and convolutional filters arise as special cases. Recall that $M_N(\Real)$ denotes the normed space of $N\times N$ real matrices with the operator norm.

\begin{Definition}\label{def:GGDS_transform}
Given a SAGS $\calA$ and a matrix-valued function $\bM_{(\cdot)} : G \in \calG_N \mapsto \bM_{G} \in M_N(\Real)$, the generalized GDS transformation $\calT$ associated with $(\calA, \bM_{(\cdot)})$ is as defined in \cref{def.GGDS_filters} with $\bF_G$ replaced by $\bM_{G}$, i.e.,
\begin{align}\label{eq.calJ-G}
\calJ : \Real^{N}\times\calG_{N}\to \Real^{N},\quad (\bx, G) \mapsto \bM_{G}\bx.
\end{align}
$\rhd\ \bx \sim \mu,\ G\mid\bx \sim \nu_{\bx},\ \calJ(\bx,G)=\bM_G\bx \sim \calT(\mu)$
\end{Definition}

For example, the GDS-FT is associated with $\bM_{G} = \bU_{G}\T$, while a generalized GDS convolutional filter corresponds to $\bM_{G} = \bF_{G}$. In both cases, the GDS transformation is defined by the pushforward of the joint distribution $\calA^{*}\mu$ through the map $\calJ: \Real^{N}\times\calG_{N}\to \Real^{N}$, which maps each graph $G$ to its associated matrix $\bM_{G}$ and then applies $\bM_{G}$ to $\bx$.

\begin{Definition}[Induced transform distribution]\label{def:ifd}
Consider a generalized GDS transformation $\calT$ associated with $(\calA, \bM_{(\cdot)})$. 
Let 
\begin{align}\label{eq.ind_trans}
\calM : \Real^N \to \calP_{p}(M_N(\Real))
\end{align}
be the map that assigns to each $\bx \in \Real^N$ the distribution of the matrix $\bM_{G}$, where $G$ is drawn from the SAGS $\nu_{\bx}(\cdot) = \calA(\bx, \cdot)$. Formally, $\calM(\bx)$ is the pushforward of $\nu_{\bx} \in \calP_{p}(\calG_N)$ through $G \mapsto \bM_{G}$. We call $\calM(\bx)$ the \emph{induced transform distribution} of $\bx$ under the generalized GDS transformation $\calT$.

\noindent$\rhd\ \bx \sim \mu,\ G\mid\bx \sim \nu_{\bx},\ \bM_G \sim \calM(\bx)$
\end{Definition}

We can now state the key continuity properties of the generalized GDS transforms $\calT$ given appropriate conditions in $\calM$. For easy reference, the maps introduced in this section are summarized in \cref{tab:aso}.

\begin{table}[!tb]
\centering
\caption{A summary of the main maps. In the last column, $\frakL(\cdot)$ is the probability distribution of its argument.} \label{tab:aso}
\centering
\begin{tabular}{cccc} 
Operator & Mapping & Eqn. & R.v.\\
\toprule
$\calA^*$ & $\calP_{p}(\mathbb{R}^N) \to \calP_{p}(\mathbb{R}^N\times \calG_N)$ & \cref{eq:cmi} & $\frakL(\bx) \mapsto \frakL(\bx,G)$ \\ 
\midrule
$\calJ$ & $\Real^N \times \calG_N \to \Real^N$ & \cref{eq.calJ-G} & $(\bx, G) \mapsto \bM_G\bx$\\
\midrule
$\calT = \calJ_*\circ \calA^*$ & $ \calP_{p}(\Real^{N}) \to \calP_{p}(\Real^{N})$ & \cref{eq.calT} & $\frakL(\bx) \mapsto \frakL(\bM_G\bx)$\\
\midrule
$\calM$ & $\Real^N \to \calP_{p}(M_N(\Real))$ & \cref{eq.ind_trans} & $\bx \mapsto \frakL(\bM_G)$ \\
\bottomrule
\end{tabular}
\end{table}  

%The following two theorems establish that the GDS transform depends continuously on the input signal distribution under mild regularity assumptions on $\Phi$. \Cref{thm:uc} considers signals supported on a compact set and requires only continuity of $\Phi$, yielding uniform continuity of $\calT_{\calG}$ over all distributions on the compact set; this setting is natural in applications where signal values are bounded, such as image pixels in $[0,255]^{N}$ or sensor readings with physical constraints. \Cref{thm:lip} removes the compactness assumption and works on the full space $\Real^{N}$, at the cost of requiring Lipschitz continuity of $\Phi$ and a finite $6$-th moment condition on the input distribution, establishing pointwise continuity of $\calT_{\calG}$.

\begin{Theorem}[Uniform continuity on compact sets]
\label{thm:uc}
Let $M_{N}(\Real)$ be endowed with the operator norm. Suppose $C \subset \Real^{N}$ is compact, and the map $\calM: \Real^N \to \calP_{p}(M_{N}(\Real))$ in \cref{def:ifd} is continuous on $C$. Then the generalized \gls{GDS} transform $\calT : \calP_{p}(C) \to \calP_{p}(\Real^{N})$ in \cref{def:GGDS_transform} is uniformly continuous.
\end{Theorem}
\begin{IEEEproof}
See \cref{proof_thm_uc} in the Appendix.
\end{IEEEproof}

In many practical settings, it is natural to assume that signals lie in a compact subset of $\Real^N$. Under this assumption, continuity of $\calM$ is a mild requirement. Specifically, the domain of $\calM$ is $C$, whereas the domain of $\calT$ is $\calP_{2}(C)$. Since $\calM$ is an intermediate object in the construction of $\calT$, establishing its continuity is generally easier. A representative example is the constant SAGS case \cite{JiTayOrt23}, where $\nu_{\bx}=\nu$ is independent of $\bx$; consequently, $\calM(\bx)=(\bM_{(\cdot)})_{*}\nu$ is constant in $\bx$ and therefore trivially continuous.

\begin{Theorem}[Pointwise continuity under Lipschitz condition]
\label{thm:lip}
Suppose the map $\calM: \Real^{N} \to \calP_{p}(M_{N}(\mathbb{R}))$ in \cref{def:ifd} is Lipschitz continuous. %, i.e., there exist $L>0$ such that for all $\bx_{1},\bx_{2} \in \Real^{N}$,
%\begin{align}\label{eq:lip}
%    W_{2}((\bT_{(\cdot)})_{\#}\mu_{\calG}(\cdot|\bx_{1}),(\bT_{(\cdot)})_{\#}\mu_{\calG}(\cdot|\bx_{2}))\leq L\|\bx_{1}-\bx_{2}\|.
%\end{align}
Then, the generalized GDS transform $\calT : \calP_{2p}(\Real^{N}) \to \calP_{p}(\Real^{N})$ in \cref{def:GGDS_transform} is continuous at any $\mu \in \calP_{2p}(\Real^{N})$ (from the $W_{2p}$ topology in the domain to the $W_p$ topology in the codomain).
\end{Theorem}
%, i.e., $\int \|\bx\|^{6} \ud \mu(\bx) < \infty$. Concretely, for all such $\mu$ and all $\varepsilon>0$, there exists $\delta>0$ (depending on $\varepsilon$ and $\mu$) such that for all $\mu'\in \calP_{p}(\Real^{N})$,
%\begin{align}
%    W_2(\mu,\mu') < \delta \implies W_{2}(\calT_{\calG}(\mu),\calT_{\calG}(\mu')) < \varepsilon.
%\end{align}

\begin{IEEEproof}
See \cref{proof_thm:lip} in the Appendix.
\end{IEEEproof}

Both finite atomic measures \cite[Theorem~6.18]{Vil09} and \glspl{GMM} \cite{DasGupta2008} are dense in $\calP_{p}(\Real^{N})$ under $W_{p}$ for every $1\leq p < \infty$. Therefore, continuity of $\calT$ in \cref{thm:uc,thm:lip} implies that its behavior on empirical measures (constructed from finitely many samples) or \glspl{GMM} closely approximates its behavior on the true underlying distribution. Consequently, for practical applications (e.g., \cref{sec:anomaly}), it is sufficient to work with empirical measures or \glspl{GMM}, which are tractable both conceptually and computationally.

Moreover, these continuity results provide a stability guarantee for GDS transforms: small perturbations of the input signal distribution arising from measurement noise, distribution shift, or finite-sample approximation, lead to controlled perturbations of the output. On compact signal domains, this stability holds with respect to $W_p$. In the noncompact case, \cref{thm:lip} establishes continuity from $W_{2p}$ to $W_{p}$, reflecting the need to control higher-order tails when the graph-dependent transform can grow with the signal magnitude. This condition covers important distribution classes, including compactly supported distributions and \glspl{GMM}.

% We also note that the finite sixth-moment condition in \cref{thm:lip} may be further relaxed; nevertheless, it already covers important cases, including compactly supported distributions and \glspl{GMM}.

\subsection{Example Application: Anomaly Detection}\label{sec:anomaly}

In classical GSP, anomaly detection applies a high-pass graph filter and flags anomalies when high-frequency GFT coefficients exceed a threshold calibrated from normal data \cite{San14b}. However, this approach disregards the distributional structure of high-frequency components, rendering it sensitive to noise and susceptible to false alarms under distribution shifts. To address this limitation, we propose a generalized GDS-based formulation that characterizes the empirical distribution of high-frequency features and detects anomalies by comparing the test distribution against a reference distribution estimated from normal data.

Consider a binary anomaly detection problem with class set $\calC=\set{\text{normal}, \text{abnormal}}$. For each class $c\in \calC$, let $\{\bx_{m}^{(c)}, G_{m}^{(c)}\}_{m=1}^{M}$ denote \gls{iid} signal-graph observations with $\bx_{m}^{(c)} \sim \mu^{(c)}$. We assume the SAGS $\calA$ is common to both classes, so that the two classes share the same signal-to-graph mapping $\bx\mapsto\nu_{\bx}$. Consequently, differences between the class-conditional joint distributions $\calA^*\mu^{(\text{normal})}$ and $\calA^*\mu^{(\text{abnormal})}$ arise solely from $\mu^{(\text{normal})} \neq \mu^{(\text{abnormal})}$. 

To extract anomaly-sensitive features, we apply the GFT to map each signal to the spectral domain. Define $\Phi:\bbR^N\times \calG_N\to \bbR^N$ by $\Phi(\bx, G)=\bU_{G}\T \bx$, where $\bU_{G}$ denotes the graph Fourier eigenbasis of $G$ with respect to a GSO $\bS_G$. Because anomalies are expected to manifest primarily in high-frequency components, let $\calI=\set{i \given \lambda_{1} \leq i \leq \lambda_{2}}$ denote a selected high-frequency index set. For each $i\in\calI$, let $\pi_i:\bbR^N\to\bbR$, $\pi_i(\bv)=v_i$, denote the coordinate projection onto the $i$-th graph Fourier coefficient. For each class $c\in\calC$, the induced high-frequency characteristic distribution is defined as
\begin{align}
\label{eq.HF_extract}
\mu_{\mathrm{HF}}^{(c)}
:= \frac{1}{|\calI|} \sum_{i\in\calI}(\pi_i\circ\Phi)\calA^*\mu^{(c)}.
\end{align}
In practice, $\mu_{\mathrm{HF}}^{(c)}$ is unknown and is estimated from training data by fitting a one-dimensional \gls{GMM} to the collected high-frequency coefficient samples. This distribution captures the statistical signature of high-frequency spectral components for class $c$. Because anomalous signals are expected to induce stronger graph-irregular variations, normal and abnormal signals may yield distinct high-frequency coefficient distributions. Consequently, $\mu_{\mathrm{HF}}^{(c)}$ serves as a discriminative distributional feature for anomaly detection. The procedure consists of the following steps.

% Strictly, $\mu_{\mathrm{HF}}^{(c)}$ is supported on $\calI\times\bbR$, i.e., a discrete set of frequency indices each carrying a continuous coefficient distribution. For tractability, we embed $\calI$ into $\bbR$ and approximate $\mu_{\mathrm{HF}}^{(c)}$ by a GMM on $\bbR^2$, treating each feature pair $(i, \widehat{\bx}_i)$ as a two-dimensional point. Intuitively, $\mu_{\mathrm{HF}}^{(c)}$ captures the distributional signature of high-frequency spectral components under class $c$, i.e., normal and abnormal signals are expected to produce characteristically different distributions, making $\mu_{\mathrm{HF}}^{(c)}$ a discriminative feature for anomaly detection. In practice, $\mu_{\mathrm{HF}}^{(c)}$ is unknown and estimated from training data via a \gls{GMM}. The anomaly detection procedure consists of the following steps.

\paragraph{Step 1: Estimate reference distribution} Assume $\{\bx_m^{(\text{normal})}\}_{m=1}^{M}$ is split into three disjoint 
sets: $\calD_{\text{train}}$ for estimating the reference distribution, $\calD_{\text{cal}}$ for threshold calibration, and $\calD_{\text{test}}$ for evaluation. We estimate the reference distribution from $\calD_{\text{train}}$ as a $K$-component GMM, i.e.,
% \begin{align*}
% P_0 := \sum_{k=1}^{K} \pi_k \calN\!\left(\bm{m}_k^{(0)}, \bSigma_k^{(0)}\right) \in \calP(\calI \times \bbR)
% \end{align*}
\begin{align*}
\hat{\mu}^{\mathrm{ref}}_{\mathrm{HF}}:= \sum_{k=1}^{K} 
\alpha_k \calN\left(m_{k}^{(\mathrm{ref})}, \left(\sigma_{k}^{(\mathrm{ref})}\right)^{2}\right) \in \calP(\Real),
\end{align*}
where $\alpha_k$, $m_k^{(\mathrm{ref})}$ and $\sigma_k^{(\mathrm{ref})}$ are estimated from high-frequency coefficient samples $\{(\pi_i\circ
\Phi)(\bx_m, G_m) \mid \bx_m\in\calD_{\text{train}},\ G_m\sim \nu_{\bx_m},\ i\in\calI\}$.

\paragraph{Step 2: Calibrate detection threshold} Partition $\calD_{\text{cal}}$ into $N_B^{\text{cal}}$ batches $\{B_b^{\text{cal}}\}_{b=1}^{N_B^{\text{cal}}}$, each of size $|B_b^{\text{cal}}|=n$. For each batch, we estimate its 
high-frequency distribution as an $L$-component GMM, i.e.,
% \begin{align*}
%     Q_b^{\text{cal}} := \sum_{\ell=1}^{L} \pi_\ell' \calN\!\left(\bm{m}_\ell^{(b)}, \bSigma_\ell^{(b)}\right) \in \calP(\calI\times\bbR),
% \end{align*}
\begin{align*}
\hat{\mu}^{\mathrm{cal},b}_{\mathrm{HF}} := \sum_{\ell=1}^{L} \alpha_{\ell}' \calN\!\left(m_\ell^{(b)}, \left(\sigma_\ell^{(b)}\right)^{2} \right) \in \calP(\Real),
\end{align*}
where $\alpha_{l}$, $m_\ell^{(b)}$ and $\sigma_\ell^{(b)}$  are estimated from  $\{(\pi_i\circ\Phi)(\bx_m, G_m) \mid \bx_m\in B_b^{\text{cal}}, G_m\sim\nu_{\bx_m},\ i\in\calI\}$. The number of components $L$ is chosen independently of $K$, as the reference and per-batch GMMs are fitted separately. The calibration score is then computed as $d_b := d( \hat{\mu}^{\mathrm{ref}}_{\mathrm{HF}}, \hat{\mu}^{\mathrm{cal},b}_{\mathrm{HF}})$, where $d(\cdot,\cdot)$ measures the discrepancy between two GMMs: either via a simple surrogate statistic based on component means, or via a more principled distributional distance (e.g., mixture Wasserstein distance). Specifically, we consider the following two choices:
\begin{itemize}
\item Peak location distance:
\begin{align}
\begin{aligned}
\label{eq.peak_loc}
d( \hat{\mu}^{\mathrm{ref}}_{\mathrm{HF}}, \hat{\mu}^{\mathrm{cal},b}_{\mathrm{HF}}) &= \frac{1}{2} \Big(\frac{1}{L}\sum_{\ell=1}^{L}\min_{k}\|m_\ell^{(b)}-m_k^{(\mathrm{ref})}\|_{2}\\
&+ \frac{1}{K}\sum_{k=1}^{K}\min_{\ell}\|m_\ell^{(b)}-m_k^{(\mathrm{ref})}\|_{2}\Big).    
\end{aligned}
\end{align}

\item Mixture Wasserstein distance:
\begin{align}
\label{eq.MWD_detec}
d(\hat{\mu}^{\mathrm{ref}}_{\mathrm{HF}}, \hat{\mu}^{\mathrm{cal},b}_{\mathrm{HF}}) = MW_{2}(\hat{\mu}^{\mathrm{ref}}_{\mathrm{HF}}, \hat{\mu}^{\mathrm{cal},b}_{\mathrm{HF}})^{2}.
\end{align}
\end{itemize}
Finally, set the detection threshold as
\begin{align*}
\tau = \mathrm{Quantile}_{1-\alpha}\!\left(\{d_b\}_{b=1}^{N_B^{\text{cal}}}\right),
\end{align*}
with significance level $\alpha$.

\paragraph{Step 3: Detection rule}
Given a test batch, we estimate its high-frequency distribution $\hat{\mu}^{\mathrm{test}}_{\mathrm{HF}}$ as in Step 2, and classify it by comparing its discrepancy from the reference against the calibrated threshold $\tau$:
\begin{align*}
\hat{c} = \begin{cases} 
\text{normal}, & d(\hat{\mu}^{\mathrm{ref}}_{\mathrm{HF}}, 
\hat{\mu}^{\mathrm{test}}_{\mathrm{HF}}) \leq \tau, \\ 
\text{abnormal}, & d(\hat{\mu}^{\mathrm{ref}}_{\mathrm{HF}}, 
\hat{\mu}^{\mathrm{test}}_{\mathrm{HF}}) > \tau.
\end{cases}
\end{align*}

We apply the scheme on a real dataset in \cref{sec:ad}.

\section{Numerical Experiments} \label{sec:experiments}

\subsection{Graph Filter Learning}\label{sec:exper_G_learning}

\subsubsection{Dataset and Experiment setup}

We evaluate both \emph{GDS-Cop} and \emph{GDS-GMMs} on a COVID-19 dataset from The New York Times.\footnote{\url{https://github.com/nytimes/covid-19-data}} We use records for California’s $N=58$ counties from July 29, 2020, to August 1, 2022, splitting them into a training set $\calD_{\mathrm{train}}$ (July 29, 2020–July 30, 2021) and a test set $\calD_{\mathrm{test}}$ (July 31, 2021–August 1, 2022). Counties are nodes with edges defined by geographic adjacency. Both training and test sets are evenly partitioned into $\calS$ non-overlapping time windows, i.e., $\calD_{\mathrm{train}} = \{ f_{\mathrm{train}}(V, \calT_s)\}_{s=1}^{\calS}$, where $f_{\mathrm{train}}(V, \calT_s) $ denotes the number of cases reports across $V$ during the $s$-th time window $\calT_s$. Each window $\calT_{s}$ consists of $\calW$ consecutive days, so that the window size is given by $\calW=\abs{\calT_{s}}$.

For \emph{GDS-Cop}, during training, the target statistics, namely mean $\bm{m}^\star$ and covariance $\bm{\Sigma}^\star$, are estimated from the subsequent window $f_{\mathrm{train}}(V, \calT_{s+1})$ following the input window $f_{\mathrm{train}}(V, \calT_{s})$. The marginal mean $m_{i}$ and variance $\sigma_{i}$ for each node $v_i$ are computed from its local observations $f_{\mathrm{train}}(v_i, \calT_s)$. After optimizing $\widetilde{\bF}_{G}$ using \cref{eq.opt}, we evaluate it on the test set. For \emph{GDS-GMMs}, by contrast, we directly estimate a joint Gaussian mixture model from the multivariate observations over the entire graph within each time window. Specifically, given the samples $f_{\mathrm{train}}(V, \calT_{s})$, we fit a joint GMM to obtain the mixture weights $\{a_{k}\}_{k=1}^{K}$, component means $\{\bm{m}_{k}\}_{k=1}^{K}$ and covariances $\{\bSigma_{k}\}_{k=1}^{K}$, where $K$ is a hyperparameter controlling the number of Gaussian components. The target distribution is estimated analogously from the subsequent window $f_{\mathrm{train}}(V, \calT_{s+1})$, yielding the parameters $\{\bm{m}^{\star}_{l},\bSigma_{l}^{\star}, b_{l}\}_{l=1}^{L}$, where $L$ is a hyperparameter controlling the number of Gaussian components. These joint GMM parameters (i.e., $\{(\bm{m}_{k},\bm{\Sigma}_{k},a_{k})\}_{k=1}^{K}$ and $\{(\bm{m}_{l}^\star, \bm{\Sigma}_{l}^\star, b_{l})\}_{l=1}^{L}$) are then used to define objective in \cref{for.WM_2} for learning the graph filter. Once the graph filter $\bF_{G}$ is optimized, its performance is evaluated on the test set.

During testing, for each test window $s = 1, \dots, \calS - 1$, the prediction is computed as
% \begin{align*}
$\widehat{f}_{\mathrm{test}}(V, \calT_{s+1}) = \widetilde{\bF}_{G} f_{\mathrm{test}}(V, \calT_s)$,
% \end{align*}
and the prediction error is evaluated using the relative squared error (RSE), defined by
\begin{align*}
\text{RSE}^{(s)}=\frac{\norm{\widehat{f}_{\mathrm{test}}(V, \calT_{s+1})-f_{\mathrm{test}}(V, \calT_{s})}_{F}^{2}}{\norm{f_{\mathrm{test}}(V, \calT_{s})}_{F}^{2}}. 
\end{align*}
The final performance is measured by the mean of RSE across all prediction windows, defined as $\text{MRSE} = \frac{1}{\calS-1}\sum_{s=1}^{\calS-1} \text{RSE}^{(s)}$.

\begin{figure*}[!t]
\centering
\includegraphics[width=\textwidth]{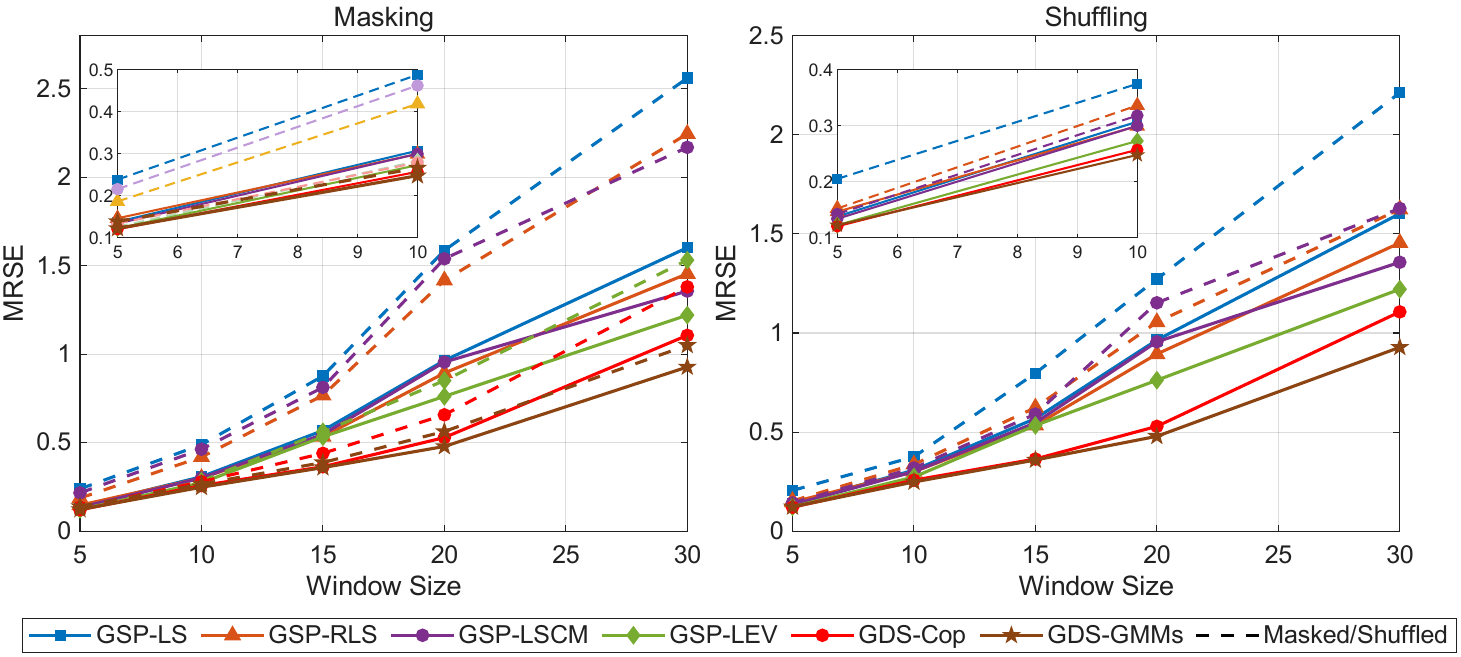}
% \vspace{-2mm}
\makebox[0.88\textwidth][c]{%
\begin{minipage}{0.3\textwidth}
\centering
\phantomsubcaption
\label{fig:masking}
{\small (a) Non-masking vs. Masking}
\end{minipage}
\hspace{0.07\textwidth}
\begin{minipage}{0.3\textwidth}
\centering
\phantomsubcaption
\label{fig:shuffling}
{\small (b) Non-shuffling vs. Shuffling}
\end{minipage}
}

\caption{MRSE as a function of window size under masking and shuffling conditions. Solid lines represent the baseline setting, while dashed lines represent the masked and shuffled variants. Each method uses the same color for corresponding dashed and solid curves. Note that GDS-based methods are invariant to shuffling.}
\label{fig:covid_RSE}
\end{figure*}

\subsubsection{Performance evaluation}
We compare both \emph{GDS-Cop} and \emph{GDS-GMMs} against the following graph filtering learning methods for graph signal prediction:
\begin{enumerate}[label=\alph*)]
\item \textbf{GSP-LS} (Least-Squares Graph Filtering) \cite{Shu13, San13}: The graph filter $\bF_{G}$ is learned by minimizing the Frobenius norm of the difference between the filtered signals and the target signals, i.e.,
$\min_{\bF_{G}} \norm{\bF_{G} \bX - \bX^\star}_F^2$.
\item \textbf{GSP-RLS} (Regularized Least-Squares Graph Filtering) \cite{Ramirez2021,Isufi2024}: The graph filter $\bF_{G}$ is learned by minimizing the same Frobenius norm objective as in GSP-LS, but with an additional $\ell_1$-norm regularization term to promote sparsity in the filter coefficients, i.e., 
% \begin{align*}
$\min_{\bF_{G}} \norm{\bF_{G} \bX - \bX^\star}_F^2 + \lambda \norm{\bF_{G}}_1$,
% \end{align*}
where $\lambda > 0$ controls the sparsity level.
\item \textbf{GSP-LSCM} (Least-Squares and Covariance-Matching Graph Filtering): This method extends GSP-LS by introducing a covariance-matching term into the objective. The graph filter $\bF_{G}$ is learned by solving:
\begin{align*}
\min_{\bF_{G}} \norm{\bF_{G} \bX - \bX^\star}_F^2 + \lambda \norm{\bF_{G} \bm{\Sigma}_{\bX} \bF_{G}\T - \bm{\Sigma}_{\bX^\star}}_F^2,
\end{align*}
where $\bm{\Sigma}_{\bX}$ and $\bm{\Sigma}_{\bX^\star}$ are the empirical covariances of the input and target signals, respectively. %and $\lambda >0$ is a trade-off parameter that balances the two terms.
\item \textbf{GSP-LEV} (Log-Evidence Maximization with Heat Kernel Mixture) \cite{Kwak2021}: The graph filter is modeled as a convex combination of $K$ heat diffusion kernels, i.e., $\bH_{G}(\calT) = \sum_{\tau \in \calT} \pi^{(\tau)} \bH_{G}(\tau)$ with $\bH_{G}(\tau) = e^{-\tau \bL_{G}}$, where $\calT$ is a predefined set of diffusion scales and $\pi^{(\tau)} \geq 0$, $\sum_{\tau\in\calT} \pi^{(\tau)} = 1$. The parameters $\pi_{\tau}$, $\alpha$ and $\gamma$, are learned by maximizing the log-evidence: 
% \red{[$\calN$ was used to denote distribution. Here, it is some function?]}
\begin{align*}
\max_{\pi, \alpha, \gamma} \log \mathcal{N}\left( \bX^{\star} \mid \bH_{G}(\calT)\bX,\, \alpha^{-1} \bI + \gamma^{-1} \bX \bX\T \right).
\end{align*}
% \cyan{where $\calN(\bm{\mu},\bSigma)$ denotes the probability density of a multivariate Gaussian distribution with mean $\bm{\mu}$ and covariance $\bSigma$.}
\end{enumerate}
For a fair comparison, the graph filter $\bF_{G}$ is parameterized as a 2-order Chebyshev polynomial with three filter coefficients, i.e., $\bF_{G}=\sum_{k=0}^{2} \btheta_{k} T(\bL_{G})$, and this parameterization is used consistently across all methods, including GSP-LS, GSP-RLS, GSP-LSCM, and the proposed GDS-Cop and GDS-GMMs.

To verify that the GDS framework does not rely on complete observations or strict temporal correspondence, we consider two stress-test settings: \emph{masking} and \emph{shuffling}. In the masking setting, each graph signal (column) is partially observed through a Bernoulli mask, where the observation probability is randomly drawn from $[0.6,0.9]$. Results are average over $10$ random masks. As shown in \cref{fig:masking}, GSP-LS, GSP-RLS, and GSP-LSCM degrade markedly due to their reliance on the complete observation assumption, while GSP-LEV, \emph{GDS-Cop} and \emph{GDS-GMMs} remain relatively robust since their distribution-matching formulations do not require full observations. In the shuffling setting, training data within each window is randomly permuted and results in \cref{fig:shuffling} are averaged over $10$ runs. From \cref{fig:shuffling}, GSP-LS, GSP-RLS, and GSP-LSCM suffer worse accuracy because they depend on strict temporal alignment, whereas GSP-LEV, \emph{GDS-Cop} and \emph{GDS-GMMs} remain stable. This resilience highlights the advantage of the GDS framework in real-world scenarios with incomplete or misaligned data, since it models distributions of graph signals rather than individual temporally aligned observations.

This advantage is also evident in the window-size study in \cref{fig:covid_RSE}. Both \emph{GDS-Cop} and \emph{GDS-GMMs} achieve strong performance across all window sizes, with one of them attaining the lowest MRSE in each case. Their benefit becomes more pronounced for larger windows. In particular, \emph{GDS-GMMs} tends to perform better when the window is large, likely because the increased number of samples allows more accurate estimation of the joint GMM. For smaller windows, \emph{GDS-Cop} is slightly more reliable, suggesting greater robustness in low-sample regimes. 

\subsubsection{When is GSP approaches better?}
We consider very small window sizes, namely $\calW=2,3,4$, with results reported in \cref{tab:MRSE_smallW}. This regime is challenging for GDS-based methods because only a few samples are available within each window to estimate the underlying distribution, leading to less reliable GMM and copula estimates. Consequently, the performance advantage of GDS-based methods diminishes, and they underperform the GSP-based baselines in this small-sample regime. These results reveal a fundamental limitation of the proposed approach: its effectiveness depends critically on having a sufficient number of observations within each window to support reliable distribution estimation.

\begin{table}[!htb]
    \centering
    \caption{MRSE under small window sizes}
    \resizebox{0.5\textwidth}{!}{%
    \begin{tabular}{c|ccc|cc}
    \toprule
     Window sizes & GSP-LS & GSP-RLS & GSP-LSCM &  GDS-Cop & GDS-GMM\\
    \midrule
      2           & 0.0487 & 0.0486 &   0.0412  &  0.0756  & 0.0943 \\
      3           & 0.0766 & 0.0762 &   0.0753  &  0.0959  & 0.1049\\
      4           & 0.1065 & 0.1061 &   0.1040  &  0.1181  & 0.1095  \\
    \bottomrule
    \end{tabular}
    }
    \label{tab:MRSE_smallW}
\end{table}

\subsection{Anomaly Detection} \label{sec:ad}
\subsubsection{Dataset and experimental setup} We conduct anomaly detection experiments on a brain ECoG dataset,\footnote{\url{https://math.bu.edu/people/kolaczyk/datasets.html}} consisting of normalized time-series recordings from $76$ brain electrodes with $4000$ time samples each, collected from an epilepsy patient during two seizure-related periods: pre-ictal (normal) and ictal (abnormal).

For each class, the ECoG time series is partitioned into non-overlapping segments of length $10$, each modeled as a graph signal on a joint sensor-time graph $G = G_0 \times H$, where $H$ is a path graph with $10$ nodes (temporal dimension) and $G_0$ is a sensor graph over the $76$ electrodes (spatial dimension). Thus $G$ has $760$ nodes indexed by $(s,t)$, yielding $400$ graph signals $\bx\in\Real^{760}$ per class, with joint Laplacian $\bL_{G}=\bL_{H}\otimes\bI_{76}+\bI_{10}\otimes\bL_{G_0}$. 

The sensor graph $G_0\sim\nu_{\bx}$ is constructed from the inter-electrode absolute correlation matrix. Given a batch of $n$ signals, each signal is reshaped into a $76\times10$ electrode-time matrix, and the resulting $10n$ temporal samples are used to estimate the $76\times76$ pairwise correlations between electrodes. The graph $G_0$ is then obtained by thresholding this matrix at a random threshold $\rho\sim\calN(\rho_0,\sigma_\rho^2)$ with $\rho_0=0.375$ and $\sigma_\rho^2=0.01$. As a practical instantiation of the SAGS, all signals within a batch share the same correlation structure, so $\nu_{\bx}$ is constant within a batch, and its randomness arises from the stochastic threshold $\rho$. The signal-adaptive nature is reflected across batches: different batches correspond to different signal realizations, inducing different graph distributions, which is precisely what distinguishes normal from abnormal conditions in the spectral domain.

For each class, the $400$ samples are split into $\calD_{\text{train}}$, $\calD_{\text{cal}}$, and $\calD_{\text{test}}$ with ratio $2{:}1{:}1$. We compute the GFT with respect to $\bL_G$ and extract the high-frequency components in the spectral range $[\lambda_1,\lambda_2]=[730,760]$, then aggregate them to estimate the class-specific reference distributions in \cref{eq.HF_extract}. \cref{fig.peak_loc} visualizes the $4$-component Gaussian mixture fits of the empirical high-frequency feature distributions for the two classes, which exhibit different peak locations in the frequency-feature space, motivating the peak location distance \cref{eq.peak_loc} as a discrepancy measure. During testing, each Monte Carlo run forms one test batch of $n$ signals drawn from a single class, with normal and abnormal batches following a $2{:}1$ ratio to reflect the class imbalance reported by NIHR, UK \cite{ratio_seizure}.

\begin{figure}[!ht]
\centering
\includegraphics[width=0.8\columnwidth, trim={3cm 9cm 3cm 9cm}, clip]{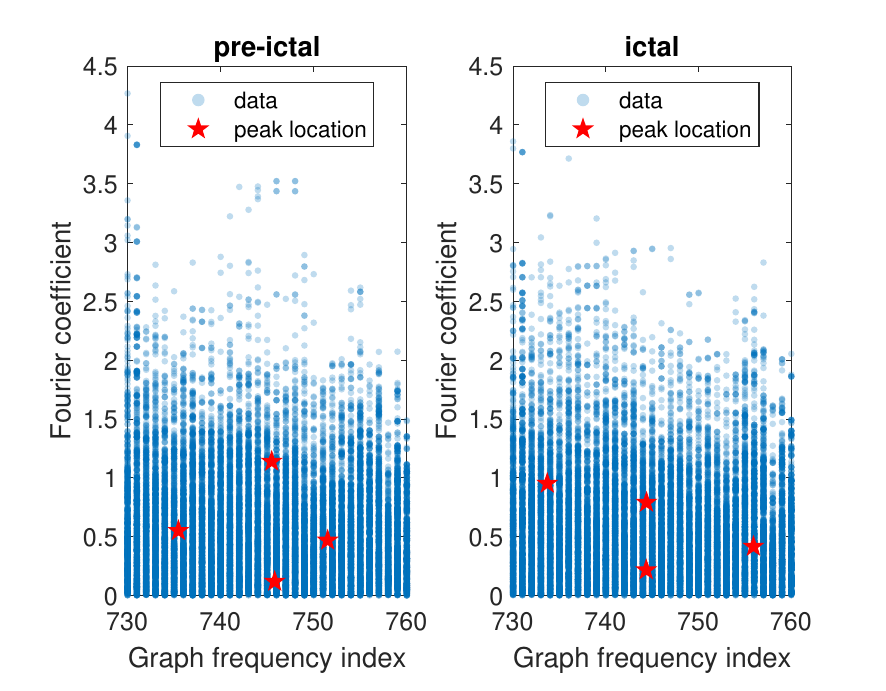}
\caption{Examples of Gaussian mixture fits of empirical distributions in the high-frequency domain. The peak locations can be used for detection.}
\label{fig.peak_loc}
\end{figure}

\begin{figure}[!ht]
\centering
\includegraphics[width=0.9\columnwidth, trim={1cm 8cm 2cm 8cm}, clip]{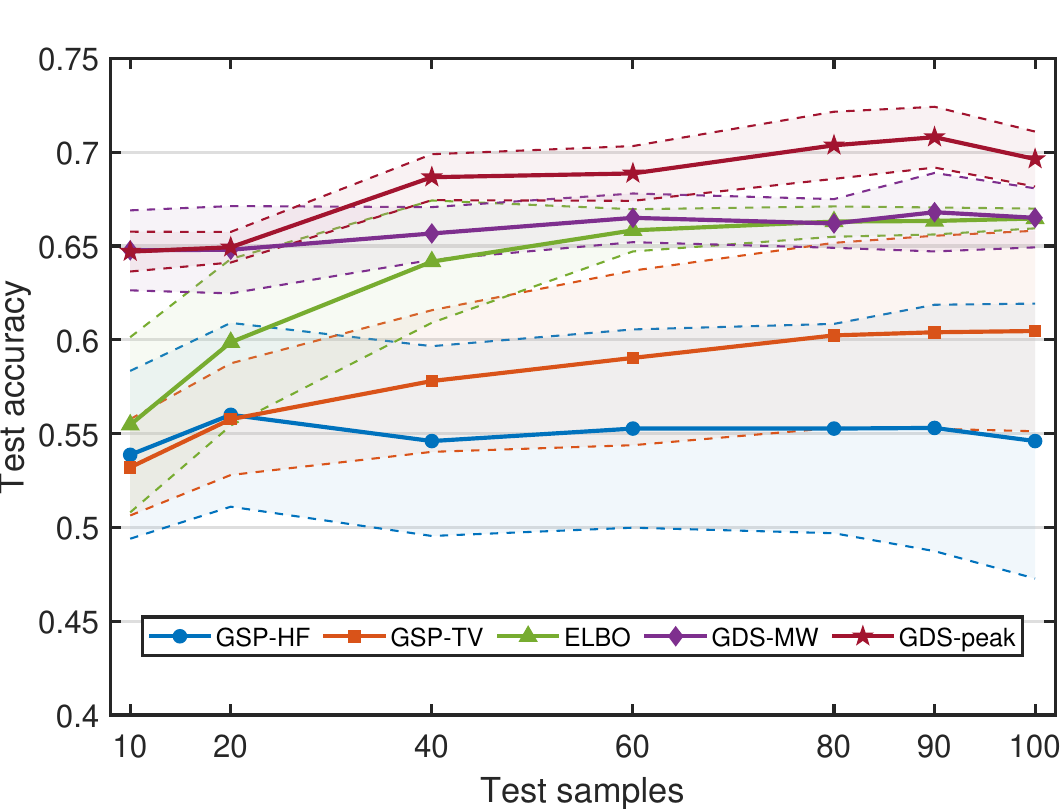}
\caption{Accuracy of anomaly detection w.r.t different sizes of test batch}
\label{fig.test_acc}
\end{figure}

\subsubsection{Performance evaluation} We compare the proposed GDS-based approach, instantiated with two discrepancy measures, namely GDS-Peak \cref{eq.peak_loc} and GDS-MW \cref{eq.MWD_detec}, against the following baseline methods for graph signal anomaly detection:

\begin{enumerate}[label=\alph*)]
\item \textbf{GSP-GFT} \cite{San14b}: The classical GFT-based detector uses the batch-averaged maximum high-frequency 
coefficient magnitude $S_b = \frac{1}{n}\sum_{\bx_m\in B_b}\max_{i\in\calI}|(\bU_{G_m}\T\bx_m)_i|$ as the test statistic, with threshold set as the average statistic over $\calD_{\text{train}}\cup\calD_{\text{cal}}$, and a test batch is declared abnormal if $S_b$ exceeds it.

\item \textbf{GSP-TV} \cite{Egilmez2014SpectralAD,SihengChen2016}: The total-variation (TV) based detector exploits the smoothness of 
normal graph signals via the batch-averaged normalized Laplacian quadratic form $\mathrm{TV}(B_b) = \frac{1}{n}\sum_{\bx_m\in B_b}
\bx_m\T\bL_{G_{m}}\bx_m/\|\bx_m\|_2^2$. Batches from $\calD_{\text{train}}\cup\calD_{\text{cal}}$ provide the median TV $m_0$ and the quantile threshold $\tau$; a test batch is declared abnormal if $|\mathrm{TV}(B_b)-m_0| > \tau$.

% \item \textbf{GSP-TV} \cite{Egilmez2014SpectralAD}: The classical total-variation (TV) based anomaly detector exploits the smoothness of normal graph signals, measured via the graph Laplacian quadratic form, as the detection statistic. The training set $\calD_{\text{train}}^{(\text{normal})}$ is partitioned into $M$ batches $\mathcal{B}_{\text{train}} = \{B_m\}_{m=1}^{M}$, where each batch $B_m = \{x_i\}_{i=1}^{\mathcal{M}_{\text{tr}}/10}$ contains $\mathcal{M}_{\text{tr}}/10$ normal samples. The per-batch TV statistic is computed as
% \begin{align*}
%     \mathrm{TV}(B_m) = \frac{1}{\mathcal{M}_{\text{tr}}/10} \sum_{i=1}^{\mathcal{M}_{\text{tr}}/10} \frac{x_i\T L_G x_i}{\|x_i\|_2^2}.
% \end{align*}
% The detection threshold is determined from the empirical distribution of $\{\mathrm{TV}(B_m)\}_{m=1}^{M}$ as
% \begin{align*}
%     \tau = \mathrm{Quantile}_{1-\alpha}\!\left(|\mathrm{TV}(B_m) - m_0|\right),
% \end{align*}
% where $m_0 = \mathrm{median}(\mathrm{TV}(B_1), \ldots, \mathrm{TV}(B_M))$ is the median TV of the training batches. Given a test batch $B_b$, the decision rule is
% \begin{align*}
%     \hat{c} = \begin{cases} \text{abnormal}, & |\mathrm{TV}(B_b) - m_0| > \tau, \\ \text{normal}, & |\mathrm{TV}(B_b) - m_0| \leq \tau. \end{cases}
% \end{align*}

\item \textbf{ELBO} \cite{an2015variational,kingma2013auto}: A variational autoencoder (VAE) \cite{kingma2013auto} is trained on 
the high-frequency GFT coefficients of normal samples from $\calD_{\text{train}}$, using the batch-averaged negative ELBO as the test statistic. The threshold $\tau$ is calibrated on $\calD_{\text{cal}}$ via the $(1-\alpha)$-quantile, and a test batch is declared abnormal if its negative ELBO exceeds $\tau$. 
\end{enumerate}

For all methods, detection performance is evaluated by the accuracy over $N_{\text{test}}$ test batches,
\begin{align}
\label{eq.acc}
\mathrm{Acc} = \frac{1}{N_{\text{test}}}\sum_{b=1}^{N_{\text{test}}} 
\mathbf{1}[\hat{c}_b = c_b],
\end{align}
where $c_b$ is the true label of the $b$-th test batch. We report $\mathrm{Acc}$ with respect to the test batch size (i.e, the number of test samples). For a fair comparison, the significance level is set to $\alpha=0.05$ for all quantile-calibrated methods.

\cref{fig.test_acc} reports the detection accuracy of all methods with respect to the test batch size $n$, from $10$ to $100$ in steps of $10$. The two distribution-based methods, namely GDS-Peak and GDS-MW, consistently achieve the highest accuracy across all batch sizes and remain robust even with few test samples, since modeling the full distribution of 
high-frequency features captures distributional shifts that point-wise statistics miss. In contrast, the point-wise baselines 
are clearly limited: GSP-GFT, relying on a single high-frequency coefficient, plateaus at a low accuracy and barely improves with larger batches, confirming that a scalar statistic is insufficient to capture distributional differences; GSP-TV and ELBO improve as the batch size increases but require many more samples to approach, yet never surpass, the GDS methods. These results demonstrate that modeling the distributional structure of high-frequency components yields more reliable and sample-efficient anomaly detection.

\section{Conclusion}
\label{sec.con}
We propose a \gls{GDS} framework that generalizes classical graph signal processing by modeling signals as probability distributions in the Wasserstein space and generalizing fixed graph topologies into signal-adaptive graph structures. This approach enables the principled handling of uncertainty and stochasticity in both signals and graphs, while strictly encompassing traditional GSP as a special case. We established a systematic correspondence between core GSP concepts and their GDS analogues, including the Fourier transform and filtering. Through two applications, graph filter learning and anomaly detection, we demonstrate the predictive benefits, robustness, and sample efficiency of the proposed framework. These results highlight the potential of the GDS framework for advancing GSP theory and for applications involving uncertain or irregular graph-structured data.

\appendix

\subsection{Proof of \cref{thm:properties_fixed}}\label{sec.propo_proof}

For any linear map $\bA\in \Real^{N\times N}$ and any coupling $\pi\in\Pi(\mu,\nu)$, the pushforward $(\bA,\bA)_{\ast}\pi$ is a coupling of $(\bA)_{\ast}\mu$ and $(\bA)_{\ast}\nu$. Therefore, 
\begin{align*} 
W_p^p\!\left((\bA)_{\ast}\mu,(\bA)_{\ast}\nu\right) &\leq \int_{\bbR^N\times\bbR^N} \|\bA\bx-\bA\by\|_2^p \,\ud\pi(\bx,\by) \\ 
% &= \int_{\bbR^N\times\bbR^N} \|\bA(\bx-\by)\|_2^p \,\mathrm{d}\pi(\bx,\by) \\ 
&\leq \|\bA\|_2^p \int_{\bbR^N\times\bbR^N} \|\bx-\by\|_2^p \,\ud\pi(\bx,\by). 
\end{align*}
Taking the infimum over all $\pi\in\Pi(\mu,\nu)$ gives
\begin{align*}
% \label{eq:linear_pushforward_lipschitz} 
W_p\!\left((\bA)_{\ast}\mu,(\bA)_{\ast}\nu\right) \leq \|\bA\|_2 W_p(\mu,\nu).
\end{align*}
For Property \ref{item:isometry_GDS_FT}, since $\bU_G$ is orthogonal, $\bU_G\T \bU_G=\bI$ and $\|\bU_G\T\bz\|_2=\|\bz\|_2$ for every $\bz\in\bbR^N$. Applying \cref{eq:linear_pushforward_lipschitz} with $\bA=\bU_G\T$ yields
\begin{align}
\label{eq.inequal_left}
W_p\!\left((\bU_G\T)_{\ast}\mu,(\bU_G\T)_{\ast}\nu\right) \leq W_p(\mu,\nu). 
\end{align}
Conversely, applying the same inequality with $\bA=\bU_G$ to the transformed measures gives
\begin{align}
\begin{aligned}
\label{eq.inequal_right}
W_p(\mu,\nu) &= W_p\!\left( (\bU_G)_{\ast}(\bU_G\T)_{\ast}\mu,\, (\bU_G)_{\ast}(\bU_G\T)_{\ast}\nu \right) \\
&\leq W_p\!\left((\bU_G\T)_{\ast}\mu,(\bU_G\T)_{\ast}\nu\right).  
\end{aligned}
\end{align}
Combining \cref{eq.inequal_left} and \cref{eq.inequal_right} proves 
\begin{align} 
W_p\!\left((\bU_G\T)_{\ast}\mu,(\bU_G\T)_{\ast}\nu\right) = W_p(\mu,\nu). 
\end{align}
% For Property \ref{item:lipschitz_GDS_filtering}, applying \cref{eq:linear_pushforward_lipschitz} with $\bA=\bF_G$ directly gives
% \begin{align} 
% W_p\!\left((\bF_G)_{\ast}\mu,(\bF_G)_{\ast}\nu\right) \leq \|\bF_G\|_2 W_p(\mu,\nu). 
% \end{align}
% In particular, since $\bF_{G}$ is a polynomial in the symmetric matrix $\bS_{G}$, it shares the same orthonormal eigenbasis as $\bS_G$ and can be written as
% \begin{align*} 
% \bF_G = \bU_G h(\bLambda_G)\bU_G\T, 
% \end{align*} 
% where $h(\cdot)$ denotes the corresponding polynomial frequency response. Hence, $\|\bF_G\|_2 = \max_i |h(\lambda_i)|$,
% so the Lipschitz constant is determined by the maximum magnitude of the graph filter frequency response.
Property \ref{item.invertibility_GDS_FT} follows from the orthogonality of $\bU_{G}$, which ensures that $(\bU_{G})_{\ast}(\bU\T_{G})_{\ast}\mu = (\bU_{G}\bU\T_{G})_{\ast}\mu = \mu$. Property \ref{item:composition} is a direct consequence of the composition rule for pushforward measures, namely $f_{\ast}(g_{\ast}\mu) = (f \circ g)_{\ast}\mu$, applied with $f = \bF_{G}$ and $g = \bH_{G}$.

\subsection{Proof of \cref{Lemma_pullback}}\label{proof_lemma1}

We begin by proving that $\calA^{\ast}\mu$ is a probability measure on $\Real^{N} \times \calG_{N}$. 
%By \cref{def.SAGS}, $\nu_{\bx} = \calA(\bx, \cdot)$ defines a probability kernel from $\Real^{N}$ to $\calG_{N}$, i.e., for every $\bx$, $\nu_{\bx}$ is a probability measure on $\calG_N$, and for every $B \in \calB(\calG_N)$, the map $\bx \mapsto \nu_{\bx}(B)$ is Borel measurable. 
Given a probability measure $\mu$ on $\Real^{N}$ and a probability kernel $\nu_{\bx}$ from $\Real^{N}$ to $\calG_{N}$, there exists a unique Borel probability measure on $\Real^{N} \times \calG_{N}$, namely the law of $(\bx, G)$ with $\bx \sim \mu$ and $G \mid \bx \sim \nu_{\bx}$ \cite[Sec.~4.1.3]{Dur19}. Setting
$B = \calG_N$ in \cref{eq:cmi} and using $\nu_{\bx}(\calG_N) = 1$ identifies the
$\Real^N$-marginal of $\calA^*\mu$ as $\mu$, i.e.,
\begin{align*}
% \label{eq:marginal}
  (\calA^*\mu)(A \times \calG_N) = \int_A \nu_{\bx}(\calG_N) \, \ud\mu(\bx) = \mu(A).  
\end{align*}
with $A \in \calB(\Real^N)$.

Next, we verify that $\calA^{\ast}\mu$ has finite $p$-th moment. Denote the metric $d_{\calG}$ on $\calG_{N}$ (i.e., the Euclidean metric induced by the adjacency parametrization) and equip $\Real^{N}\times\calG_{N}$ with 
\begin{align*}
% \label{eq:prod-metric}
    d_{\times}\big((\bx, G), (\by, H)\big) := \big(\|\bx - \by\|^p + d_{\calG}(G, H)^p\big)^{1/p}.
\end{align*}
Fix $z_0 = (\mathbf{0}, G_0)$. For any $ z = (\bx,G) \in \Real^{N} \times \calG_{N}$, we have 
\begin{align*}
    d_{\times}(z,z_{0})^{p} =  \| \bx \|^{p} + d_{\calG}(G, G_{0})^{p} \geq 0.
\end{align*}
Tonelli's theorem \cite[Section 1.7]{Dur19} and the definition of $\calA^{*}\mu$ yield
\begin{align*}
\begin{aligned}
&\int_{\Real^{N} \times \calG_{N}} d_{\times}(z, z_{0})^{p} \ud(\calA^{\ast}\mu)(z) \\
   &\quad = \int_{\Real^{N}}\int_{\calG_{N}} \left( \|\bx\|^{p} + d_{\calG}(G,G_{0})^{p} \right) \ud \nu_{\bx}(G) \ud\mu(\bx)\\
   &\quad = \int_{\Real^{N}} \|\bx\|^{p} \ud \mu(\bx) + \int_{\Real^{N}}\int_{\calG_{N}}d_{\calG}(G,G_{0})^{p}\ud \nu_{\bx}(G)\ud\mu(\bx).
\end{aligned}
\end{align*}
The first term is finite since $\mu \in \calP_{p}(\Real^{N})$. For the second term, since $\delta_{G_{0}}$ is the Dirac measure at $G_{0}$, the only coupling between $\nu_{\bx}$ and $\delta_{G_0}$ is the product $\nu_{\bx} \otimes
\delta_{G_0}$. Hence, for every $\bx \in \Real^N$, $W_p^p(\nu_{\bx}, \delta_{G_0}) = \int_{\calG_N} d_{\calG}(G, G_0)^p \, \ud\nu_{\bx}(G)$.
% \begin{align}
% \label{eq:wp-dirac}    
% \end{align}
Therefore, 
\begin{align*}
    &\int_{\Real^{N}}\int_{\calG_{N}}d_{\calG}(G,G_{0})^p \ud\nu_{\bx}(G)\ud\mu(\bx)\\
    & = \int_{\Real^{N}} W_{p}^p(\nu_{\bx},\delta_{G_0}) \ud\mu(\bx) < \infty,
\end{align*}
where the last inequality follows from \cref{eq:sags-moment}. Combining this with the finiteness of the first term gives 
\begin{align*}
    \int_{\Real^{N}\times\calG_N} d_{\times}(z,z_0)^p \ud(\calA^*\mu)(z)<\infty.
\end{align*}
Thus $\calA^{\ast}\mu$ has finite $p$-th moment on $\Real^N\times\calG_N$. Since $\calA^{\ast}\mu$ is a Borel probability measure on $\Real^{N}\times\calG_{N}$, we conclude that $\calA^{\ast}\mu \in \calP_{p}(\Real^{N}\times \calG_{N})$.

\subsection{Proof of \cref{thm:uc}}\label{proof_thm_uc}

Since $C$ is compact, every probability measure supported on $C$ has finite $p$-th moment, and hence $\calP(C)=\calP_{p}(C)$. By Prokhorov's theorem \cite{Bill1999} and the equivalence between weak convergence and $W_{p}$ convergence on compact metric spaces, $\calP(C)$ is compact under $W_{p}$. Since $\calM: C \to \calP_{p}(M_{N}(\Real))$ is continuous and $C$ is compact, the Heine-Cantor theorem implies that $\calM$ is uniformly continuous on $C$. Define
\begin{align}
\label{eq.def_S}
    R &:= \sup_{\bx \in C} \|\bx\| < \infty, \\ 
    S &:= \sup_{\bx \in C} W_{p}(\calM(\bx), \delta_{\bm{0}}) < \infty, 
\end{align}
where $\delta_{\bm{0}}$ denotes the Dirac measure at the zero matrix in $M_{N}(\Real)$. 

\emph{Step 1 (Dirac measures):} We first show that the map $\bx \mapsto \calT(\delta_{\bx})$ is uniformly continuous from $C$ to $\calP_{p}(\Real^{N})$. Fix $\bx_{1}, \bx_{2} \in C$. Let $\gamma_{\bx_{1},\bx_{2}}$ be an optimal coupling between $\calM(\bx_{1})$ and $\calM(\bx_{2})$, i.e.,
\begin{align}
\begin{aligned}
W_{p}^{p}(&\calM(\bx_{1}),\calM(\bx_{2}))\\
&=\int \|\bM_{1}-\bM_{2}\|^{p} \ud \gamma_{\bx_{1},\bx_{2}}(\bM_{1},\bM_{2}).   
\end{aligned}
\end{align}
Since $\calM$ is continuous on the compact set $C$, the map $(\bx_1,\bx_2)\mapsto(\calM(\bx_1),\calM(\bx_2))$ is continuous, hence Borel measurable, from $C\times C$ to $\calP_p(M_N(\Real))\times\calP_p(M_N(\Real))$. By the measurable selection theorem for optimal couplings \cite[Corollary~5.22]{Vil09}, the fiberwise optimal couplings can be chosen so that $(\bx_1,\bx_2)\mapsto\gamma_{\bx_1,\bx_2}$ is a Borel-measurable kernel from $C\times C$ to $\calP\big(M_N(\Real)\times M_N(\Real)\big)$; we fix such a measurable selection throughout. Push $\gamma_{\bx_{1},\bx_{2}}$ forward through the map $(\bM_{1},\bM_{2}) \mapsto (\bM_{1}\bx_{1},\bM_{2}\bx_{2})$ and denote the resulting measure on $\Real^{N} \times \Real^{N}$ by $\gamma'_{\bx_{1},\bx_{2}}$. 
By construction, $\gamma'_{\bx_{1},\bx_{2}}$ is a coupling between $\calT(\delta_{\bx_{1}})$ and $\calT(\delta_{\bx_{2}})$, and, since pushforward under a fixed continuous map preserves measurability of kernels, $(\bx_1,\bx_2)\mapsto\gamma'_{\bx_1,\bx_2}$ is likewise a Borel-measurable kernel on $C\times C$. Hence,
\begin{align}
\begin{aligned}
\label{proof.step1}
W_{p}^{p}&(\calT(\delta_{\bx_{1}}),\calT(\delta_{\bx_{2}}))\\
&\leq \int \|\bM_{1}\bx_{1} - \bM_{2}\bx_{2}\|^{p} \ud \gamma_{\bx_{1},\bx_{2}}(\bM_{1},\bM_{2}). 
\end{aligned}
\end{align}
Using $\bM_{1}\bx_{1}-\bM_{2}\bx_{2}=(\bM_{1}-\bM_{2})\bx_{1}+\bM_{2}(\bx_{1}-\bx_{2})$ and $\|a+b\|^{p}\leq 2^{p-1}(\|a\|^{p}+\|b\|^{p})$, we obtain
\begin{align}\label{ineq:Wpp}
\ml{
W_{p}^{p}(\calT(\delta_{\bx_{1}}),\calT(\delta_{\bx_{2}})) \leq 
2^{p-1} \int \|(\bM_{1}-\bM_{2})\bx_{1}\|^{p} \ud\gamma_{\bx_{1},\bx_{2}} \\
+ 2^{p-1} \int \|\bM_{2}(\bx_{1}-\bx_{2})\|^{p} \ud\gamma_{\bx_{1},\bx_{2}}.
}
\end{align}
The first term on the \gls{RHS} of \cref{ineq:Wpp} is bounded as:
\begin{align}\label{proof.first_step1}
\int \|(\bM_{1}-\bM_{2})\bx_{1}\|^{p} \ud\gamma_{\bx_{1},\bx_{2}} 
\leq R^{p} W_{p}^{p}( \calM(\bx_{1}), \calM(\bx_{2}) ). 
\end{align}
For the second term, using $\|\bM_{2}(\bx_{1}-\bx_{2})\| \leq \|\bM_{2}\|\|\bx_{1}-\bx_{2}\|$ and noting that the second marginal of $\gamma_{\bx_{1},\bx_{2}}$ is $\calM(\bx_{2})$, we have
\begin{align}
&\int \|\bM_{2}(\bx_{1}-\bx_{2})\|^{p} \ud\gamma_{\bx_{1},\bx_{2}} \nn
&\leq \|\bx_{1}-\bx_{2}\|^{p} \int \|\bM_{2}\|^{p} \ud \calM(\bx_{2})(\bM_{2})\nn
&= \|\bx_{1}-\bx_{2}\|^{p}W_{p}^{p}(\calM(\bx_{2}),\delta_{\bm{0}})\nn
&\leq S^{p} \|\bx_{1}-\bx_{2}\|^{p}. \label{proof.sec_step1}   
\end{align}
Combining \cref{proof.step1,proof.first_step1,proof.sec_step1}, we arrive at
\begin{align}
\begin{aligned}
\label{Proof.step1_final}
W_{p}^{p}(\calT(\delta_{\bx_{1}}),\calT(\delta_{\bx_{2}}))
&\leq 2^{p-1}R^{p} W_{p}^{p}(\calM(\bx_{1}), \calM(\bx_{2})) \\
& \quad + 2^{p-1} S^{p} \|\bx_{1}-\bx_{2}\|^{p}. 
\end{aligned}
\end{align}
Since $\calM$ is uniformly continuous on $C$, both terms on the \gls{RHS} converge to zero uniformly as $\|\bx_{1}-\bx_{2}\| \to 0$. Hence, $\bx \mapsto \calT(\delta_{\bx})$ is uniformly continuous on $C$.
% From \cref{Proof.step1_final}, we extract the following quantitative bounds. For every $\varepsilon>0$, there exist constants $B_{\varepsilon}, C_{\varepsilon}>0$ such that for all $\bx_{1}, \bx_{2}\in K$, we have
% \begin{align}
% \begin{aligned}
% \label{proof.step1_rig_bound}
% &\|\bx_{1}-\bx_{2}\| \geq \varepsilon \\
%     &\qquad \implies W_{2}(\calT_{\calG}(\delta_{\bx_{1}}),\calT_{\calG}(\delta_{\bx_{2}}))^{2}\leq B_{\varepsilon}\|\bx_{1}-\bx_{2}\|^{2}
% \end{aligned}
% \end{align}
% and 
% \begin{align}
% \label{proof.step1_left_bound}
%     \|\bx_{1}-\bx_{2}\| < \varepsilon \implies W_{2}(\calT_{\calG}(\delta_{\bx_{1}}),\calT_{\calG}(\delta_{\bx_{2}}))^{2} \leq C_{\varepsilon}
% \end{align}
% where $C_{\varepsilon} \to 0 $ as $\varepsilon \to 0$. 

\emph{Step 2 (General measures)}: Let $\mu, \mu' \in \calP(C)$. Let $\eta$ be an optimal coupling between $\mu$ and $\mu'$. Thus,
\begin{align}
W_{p}^{p}(\mu,\mu')=\int_{ C \times C} \|\bx_{1}-\bx_{2}\|^{p} \ud \eta(\bx_{1},\bx_{2}).
\end{align}
For each $(\bx_{1}, \bx_{2})\in C \times C$, let $\gamma'_{\bx_{1},\bx_{2}}$ be the coupling
between $\calT(\delta_{\bx_{1}})$ and $\calT(\delta_{\bx_{2}})$ constructed in \emph{Step 1}, i.e. $\gamma'_{\bx_{1},\bx_{2}}$ is obtained by pushing forward $\gamma_{\bx_{1},\bx_{2}}$ through the map $(\bM_{1},\bM_{2}) \mapsto (\bM_{1}\bx_{1},\bM_{2}\bx_{2})$, using the Borel-measurable kernel $(\bx_1,\bx_2)\mapsto\gamma_{\bx_1,\bx_2}$ fixed in \emph{Step 1}. Define a measure $\eta'$ on $\Real^{N} \times \Real^{N}$ by requiring that, for every bounded continuous function $g: \Real^{N} \times \Real^{N} \to \Real$,
\begin{align}
\begin{aligned}
&\int_{\Real^{N}\times\Real^{N}}g(\bw_{1},\bw_{2})\ud\eta'(\bw_{1},\bw_{2})\\
    &:= \int_{C \times C}\int_{\Real^{N}\times\Real^{N}} g(\bw_{1},\bw_{2}) \ud\gamma'_{\bx_{1},\bx_{2}}(\bw_{1},\bw_{2})\ud\eta(\bx_{1},\bx_{2}).
\end{aligned}
\end{align}
This is well defined since $(\bx_1,\bx_2)\mapsto\gamma'_{\bx_1,\bx_2}$ is a Borel-measurable kernel and $g$ is bounded and continuous, the map $(\bx_1,\bx_2)\mapsto\int g\,d\gamma'_{\bx_1,\bx_2}$ is Borel measurable and bounded on the compact set $C\times C$. Therefore, the outer integral against $\eta$ is well defined; the resulting functional $g\mapsto\int g\,d\eta'$ is positive and linear on bounded continuous functions with $\int \ud\eta'=1$, so by the Riesz representation theorem it defines a (unique) Borel probability measure $\eta'$ on $\Real^N\times\Real^N$. We now verify that $\eta'$ is a coupling between $\calT(\mu)$ and $\calT(\mu')$. Let $h: \Real^{N} \to \Real$ be bounded and continuous. Since the first marginal of $\gamma'_{\bx_{1},\bx_{2}}$ is $\calT(\delta_{\bx_{1}})$, we have
\begin{align}
\begin{aligned}
&\int_{\Real^{N}\times\Real^{N}} h(\bw_{1})\ud\eta'(\bw_{1},\bw_{2}) \\
   & = \int_{C\times C}\int_{\Real^{N}\times\Real^{N}} h(\bw_{1}) \ud\gamma'_{\bx_{1},\bx_{2}}(\bw_{1},\bw_{2})\ud\eta(\bx_{1},\bx_{2})\\
   & = \int_{C\times C}\int_{\Real^{N}} h(\bw_{1}) \ud\calT(\delta_{\bx_{1}})(\bw_{1})\ud\eta(\bx_{1},\bx_{2})\\
   & = \int_{C}\int_{\Real^{N}} h(\bw_{1}) \ud\calT(\delta_{\bx_{1}})(\bw_{1})\ud\mu(\bx_{1})\\
   & = \int_{\Real^{N}} h(\bw_{1})\ud \calT(\mu)(\bw_{1}).
\end{aligned}
\end{align}
Thus the first marginal of $\eta'$ is $\calT(\mu)$. By the same argument, the second marginal of $\eta'$ is $\calT(\mu')$. Since $\eta'$ is a coupling between $\calT(\mu)$ and $\calT(\mu')$, 
\begin{align}
\begin{aligned}
\label{proof.step2_bound}
    &W_{p}^{p}(\calT(\mu),\calT(\mu'))\leq \int_{\Real^{N}\times\Real^{N}} \|\bw_{1}-\bw_{2}\|^{p}\ud\eta'(\bw_{1},\bw_{2})\\
    & = \int_{C \times C}\int_{\Real^{N} \times \Real^{N}} \|\bw_{1}-\bw_{2}\|^{p} \ud\gamma'_{\bx_{1},\bx_{2}}(\bw_{1},\bw_{2})\ud\eta(\bx_{1},\bx_{2})
\end{aligned}
\end{align}
Let $\omega(r) := \sup_{\substack{\bx_1,\bx_2\in C\\ \|\bx_1-\bx_2\|\le r}} W_{p}\left(\calM(\bx_1),\calM(\bx_2)\right)$. Since $\calM$ is uniformly continuous on $C$, we have $\omega(r) \to 0$ as $r \to 0$. Moreover, by the estimate in \emph{Step 1}, the inner integral satisfies 
\begin{align*}
\begin{aligned}
&\int_{\Real^{N}\times \Real^{N}} \|\bw_{1}-\bw_{2}\|^{p} \ud\gamma'_{\bx_{1},\bx_{2}}(\bw_{1},\bw_{2}) \\
&\leq 2^{p-1}R^{p} W_{p}^{p}(\calM(\bx_{1}), \calM(\bx_{2})) + 2^{p-1} S^{p} \|\bx_{1}-\bx_{2}\|^{p}.  
\end{aligned}
\end{align*}
Fix $\varepsilon>0$. Split the integral over $ C \times C$ into the regions $\|\bx_{1}-\bx_{2}\|<\varepsilon$ and $\|\bx_{1}-\bx_{2}\| \geq \varepsilon$. On the first region, we have $W_{p}^{p}(\calM(\bx_{1}), \calM(\bx_{2}))\leq \omega(\varepsilon)^{p}$. On the second region, using the triangle inequality and the definition of $S$ in \cref{eq.def_S}, we have
\begin{align*}
& W_{p}(\calM(\bx_{1}), \calM(\bx_{2})) \\
& \leq W_{p}(\calM(\bx_{1}), \delta_{\bm{0}}) +  W_{p}(\calM(\bx_{2}), \delta_{\bm{0}}) \leq 2S,
\end{align*}
so that
\begin{align*}
\begin{aligned}
      W_{p}^{p}(\calM(\bx_{1}), \calM(\bx_{2})) \leq (2S)^{p}.
\end{aligned}
\end{align*}
Moreover, by Markov's inequality, we have
\begin{align*}
\begin{aligned}
    &\eta(\|\bx_{1}-\bx_{2}\|\geq \varepsilon) \\
    &\quad \leq \frac{1}{\varepsilon^{p}} \int_{ C \times C} \|\bx_{1}-\bx_{2}\|^{p} \ud\eta(\bx_{1},\bx_{2}) = \frac{W_{p}^{p}(\mu,\mu')}{\varepsilon^{p}}
\end{aligned}
\end{align*}
where the last equality follows from the optimality of the coupling $\eta$. Therefore, 
\begin{align*}
    W_{p}^{p}&(\calT(\mu),\calT(\mu')) \leq 2^{p-1}R^{p}\omega(\varepsilon)^{p}\\
    &+2^{p-1}R^{p}(2S)^{p}\frac{W_{p}^{p}(\mu,\mu')}{\varepsilon^{p}}+2^{p-1}S^{p}W_{p}^{p}(\mu,\mu').
\end{align*}
Given any $\alpha>0$, choose $\varepsilon > 0$ sufficiently small such that $2^{p-1}R^{p}\omega(\varepsilon)^{p} < \alpha^{p}/2$. For this fixed $\varepsilon$, define $B_\varepsilon:= 2^{p-1}R^p(2S)^p/\varepsilon^p + 2^{p-1}S^p$. Choose $\kappa:=\left(\alpha^{p}/(2B_{\varepsilon})\right)^{1/p}$. Then, for all $\mu,\mu' \in \calP_{p}(C)$ satisfying $W_{p}(\mu,\mu') < \kappa $, we have 
\begin{align}
\begin{aligned}
W_{p}^{p}(\calT(\mu),\calT(\mu')) 
  & < 2^{p-1}R^{p}\omega(\varepsilon)^{p} + B_{\varepsilon}W_{p}^{p}(\mu,\mu') \\
  & < \frac{\alpha^{p}}{2} + B_{\varepsilon}\kappa^{p}
   = \frac{\alpha^{p}}{2} + \frac{\alpha^{p}}{2} = \alpha^{p}.
\end{aligned}
\end{align}
Since $\kappa$ depends only on $\alpha$ and not on the particular $\mu$ and $\mu'$, this proves that $\calT: \calP_{p}(C) \to \calP_{p}(\Real^{N})$ is uniformly continuous.

\subsection{Proof of \cref{thm:lip}}\label{proof_thm:lip}

Let $L>0$ be the Lipschitz constant of $\calM$, and let $\delta_{\bm{0}}$ denote the Dirac measure at the zero matrix in $M_{N}(\Real)$. Since $\calM(\bm{0})\in \calP_{p}(M_{N}(\Real))$, define 
\begin{align*}
    S_{0}:=W_{p}(\calM(\bm{0}),\delta_{\bm{0}})< \infty.
\end{align*}
By the triangle inequality and the Lipschitz continuity of $\calM$, for every $\bx \in \Real^{N}$, we have
\begin{align}
\begin{aligned}
\label{eq:M_growth}
W_{p}(\calM(\bx),\delta_{\bm{0}})&\leq W_{p}(\calM(\bx),\calM(\bm{0})) + W_{p}(\calM(\bm{0}),\delta_{\bm{0}})\\
& \leq L \|\bx\| + S_{0}
\end{aligned}
\end{align}
We first verify that $\calT(\mu)\in\calP_p(\Real^N)$ for every $\mu\in\calP_{2p}(\Real^N)$. Let $\mu\in\calP_{2p}(\Real^N)$. Since $\calT(\mu)$ is the distribution of $\bw=\bM\bx$, where $\bx\sim\mu$ and $\bM\sim\calM(\bx)$, we have
\begin{align}
\int_{\Real^N}&\|\bw\|^p\,\ud\calT(\mu)(\bw)\nn
& =
\int_{\Real^N}\int_{M_N(\Real)}
\|\bM\bx\|^p
\,\ud\calM(\bx)(\bM)\,\ud\mu(\bx) \nn
& \leq
\int_{\Real^N}\|\bx\|^p
\int_{M_N(\Real)}
\|\bM\|^p \ud\calM(\bx)(\bM)\ud\mu(\bx) \nn
&=
\int_{\Real^N}
\|\bx\|^p
W_p^p\big(\calM(\bx),\delta_{\mathbf 0}\big) \ud\mu(\bx) \nn
&\leq
\int_{\Real^N} \|\bx\|^p \big(L\|\bx\|+S_0\big)^p \ud\mu(\bx). \label{eq:T_well_defined}
\end{align}
Moreover, 
\begin{align*}
\|\bx\|^p\big(L\|\bx\|+S_0\big)^p &\leq 2^{p-1}L^p\|\bx\|^{2p} + 2^{p-1}S_0^p\|\bx\|^p.
\end{align*}
Since $\mu\in\calP_{2p}(\Real^N)$, the right-hand side of \cref{eq:T_well_defined} is finite. Hence, $\calT(\mu)\in\calP_p(\Real^N)$. Therefore, $\calT$ is well-defined on $\calP_{2p}(\Real^N)$.

We now prove continuity at an arbitrary $\mu\in\calP_{2p}(\Real^N)$. Let $\mu'\in\calP_{2p}(\Real^N)$, and let $\eta$ be an optimal coupling between $\mu$ and $\mu'$ for $W_{2p}$. Then,
\begin{align}
\label{eq:eta_w2p}
    W_{2p}(\mu,\mu') = \left(\int_{\Real^N\times\Real^N}
\|\bx_1-\bx_2\|^{2p} \ud\eta(\bx_1,\bx_2) \right)^{\frac{1}{2p}}.
\end{align}
For each pair $(\bx_1,\bx_2)$, let $\gamma_{\bx_1,\bx_2}$ be an optimal coupling between $\calM(\bx_1)$ and $\calM(\bx_2)$. Since $\calM$ is Lipschitz continuous, it is in particular Borel measurable from $\Real^N$ to $\calP_p(M_N(\Real))$. Using the same argument as in the proof of \cref{thm:uc}, we can fix a Borel-measurable kernel $(\bx_1,\bx_2)\mapsto\gamma_{\bx_1,\bx_2}$ from $\Real^N\times\Real^N$ to $\calP(M_N(\Real)\times M_N(\Real))$ such that $\gamma_{\bx_1,\bx_2}$ is optimal between $\calM(\bx_1)$ and $\calM(\bx_2)$ for every $(\bx_1,\bx_2)$.
Push it forward through the map
\begin{align*}
(\bM_1,\bM_2)&\mapsto(\bM_1\bx_1,\bM_2\bx_2),
\end{align*}
and denote the resulting coupling between $\calT(\delta_{\bx_{1}})$ and $\calT(\delta_{\bx_{2}})$ by $\gamma'_{\bx_{1},\bx_{2}}$. Since pushforward under a fixed continuous map preserves measurability of kernels, $(\bx_1,\bx_2)\mapsto\gamma'_{\bx_1,\bx_2}$ is likewise a Borel-measurable kernel. Define the corresponding fiberwise cost
\begin{align*}
Q_p(\bx_1,\bx_2)
&:=
\int_{\Real^N\times\Real^N}
\|\bw_1-\bw_2\|^p
\,\ud\gamma'_{\bx_1,\bx_2}(\bw_1,\bw_2).
\end{align*}
By construction,
\begin{align*}
& Q_p(\bx_1,\bx_2) \\
& \quad = \int_{M_{N}(\Real)\times M_{N}(\Real)}
\|\bM_1\bx_1-\bM_2\bx_2\|^p \ud\gamma_{\bx_1,\bx_2}(\bM_1,\bM_2). 
\end{align*}
Using the decomposition
\begin{align*}
\bM_1\bx_1-\bM_2\bx_2
&=
(\bM_1-\bM_2)\bx_1+\bM_2(\bx_1-\bx_2)
\end{align*}
and the triangle inequality
\begin{align*}
\|a+b\|^p&\leq 2^{p-1}\big(\|a\|^p+\|b\|^p\big),
\end{align*}
we obtain
\begin{align}
\label{eq:Q_split}
\begin{aligned}
&Q_p(\bx_1,\bx_2)\\
&\quad\leq
2^{p-1}
\int_{M_{N}(\Real)\times M_{N}(\Real)}
\|(\bM_1-\bM_2)\bx_1\|^p
\,\ud\gamma_{\bx_1,\bx_2} \\
&\qquad+
2^{p-1}
\int_{M_{N}(\Real)\times M_{N}(\Real)}
\|\bM_2(\bx_1-\bx_2)\|^p
\,\ud\gamma_{\bx_1,\bx_2}.
\end{aligned}
\end{align}
For the first term, since
\begin{align*}
\|(\bM_1-\bM_2)\bx_1\|
&\leq
\|\bM_1-\bM_2\|_{\mathrm{op}}\|\bx_1\|,
\end{align*}
and $\gamma_{\bx_1,\bx_2}$ is optimal between $\calM(\bx_1)$ and $\calM(\bx_2)$, we have
\begin{align}
\label{eq:first_Q_bound}
\begin{aligned}
&\int_{M_{N}(\Real)\times M_{N}(\Real)}
\|(\bM_1-\bM_2)\bx_1\|^p
\,\ud\gamma_{\bx_1,\bx_2}\\
&\qquad \qquad \qquad \leq
\|\bx_1\|^p
W_p^p\big(\calM(\bx_1),\calM(\bx_2)\big) \\
&\qquad \qquad \qquad \leq
L^p\|\bx_1\|^p\|\bx_1-\bx_2\|^p .
\end{aligned}
\end{align}
For the second term, since
\begin{align*}
\|\bM_2(\bx_1-\bx_2)\|
&\leq
\|\bM_2\|_{\mathrm{op}}\|\bx_1-\bx_2\|,
\end{align*}
and the second marginal of $\gamma_{\bx_1,\bx_2}$ is $\calM(\bx_2)$, we obtain
\begin{align}
\label{eq:second_Q_bound}
\begin{aligned}
\int_{M_{N}(\Real)\times M_{N}(\Real)}
&\|\bM_2(\bx_1-\bx_2)\|^p
\,\ud\gamma_{\bx_1,\bx_2}\\
&\qquad \leq
\|\bx_1-\bx_2\|^p
W_p^p\big(\calM(\bx_2),\delta_{\mathbf 0}\big) \\
&\qquad \leq
\big(L\|\bx_2\|+S_0\big)^p
\|\bx_1-\bx_2\|^p ,
\end{aligned}
\end{align}
where the last inequality follows from \cref{eq:M_growth}. Combining
\cref{eq:Q_split,eq:first_Q_bound,eq:second_Q_bound}, we get
\begin{align}
\begin{aligned}
\label{eq:Q_final_bound}
&Q_p(\bx_1,\bx_2)\\
&\qquad \leq
2^{p-1}
\Big[
L^p\|\bx_1\|^p
+
\big(L\|\bx_2\|+S_0\big)^p
\Big]
\|\bx_1-\bx_2\|^p .   
\end{aligned}
\end{align}

Next, because $(\bx_1,\bx_2)\mapsto\gamma'_{\bx_1,\bx_2}$ is a Borel-measurable kernel, we may define a measure $\eta'$ on $\Real^N\times\Real^N$ via the disintegration formula: for any Borel set $A\subseteq\Real^N\times\Real^N$,
\begin{align*}
\eta'(A) := \int_{\Real^N\times\Real^N} \gamma'_{\bx_1,\bx_2}(A)\,\ud\eta(\bx_1,\bx_2),
\end{align*}
i.e., $\eta'$ is obtained by first sampling $(\bx_1,\bx_2)\sim\eta$ and then sampling $(\bw_1,\bw_2)\sim\gamma'_{\bx_1,\bx_2}$. Since the first marginal of $\gamma'_{\bx_1,\bx_2}$ equals $\calT(\delta_{\bx_1})$ for every $\bx_1$, and the first marginal of $\eta$ is $\mu$, integrating over $\bx_1\sim\mu$ shows that the first marginal of $\eta'$ is $\int\calT(\delta_{\bx_1})\,\ud\mu(\bx_1)=\calT(\mu)$; by the symmetric argument, the second marginal of $\eta'$ is $\calT(\mu')$. Hence $\eta'$ is a valid coupling of $\calT(\mu)$ and $\calT(\mu')$. Moreover, since $\|\bw_1-\bw_2\|^p\geq 0$, Tonelli's theorem gives
\begin{align*}
&\int_{\Real^N\times\Real^N} \|\bw_1-\bw_2\|^p\,\ud\eta'(\bw_1,\bw_2) \\
&= \int_{\Real^N\times\Real^N} Q_p(\bx_1,\bx_2)\,\ud\eta(\bx_1,\bx_2).
\end{align*}
Since $\eta'$ is a coupling of $\calT(\mu)$ and $\calT(\mu')$, we conclude
\begin{align*}
W_p^p\big(\calT(\mu),\calT(\mu')\big)
&\leq
\int_{\Real^{N}\times\Real^{N}} Q_p(\bx_1,\bx_2)\,\ud\eta(\bx_1,\bx_2).
\end{align*}
Using \cref{eq:Q_final_bound}, this gives
\begin{align}
\label{eq:T_continuity_pre_holder}
\begin{aligned}
&W_p^p\big(\calT(\mu),\calT(\mu')\big)\\
&\qquad\leq
2^{p-1}L^p
\int
\|\bx_1\|^p\|\bx_1-\bx_2\|^p
\,\ud\eta \\
&\qquad\qquad +
2^{p-1}
\int
\big(L\|\bx_2\|+S_0\big)^p
\|\bx_1-\bx_2\|^p
\,\ud\eta .
\end{aligned}
\end{align}
By Hölder's inequality,
\begin{align}
\label{eq:holder_first}
\begin{aligned}
\int&
\|\bx_1\|^p\|\bx_1-\bx_2\|^p
\,\ud\eta \\ 
&\leq
\left(
\int \|\bx_1\|^{2p}\,\ud\eta
\right)^{1/2}
\left(
\int \|\bx_1-\bx_2\|^{2p}\,\ud\eta
\right)^{1/2} \\
&=
W_{2p}^p(\mu,\delta_{\mathbf 0})
W_{2p}^p(\mu,\mu').
\end{aligned}
\end{align}
Similarly,
\begin{align}
\label{eq:holder_second}
\begin{aligned}
&\int
\big(L\|\bx_2\|+S_0\big)^p
\|\bx_1-\bx_2\|^p
\,\ud\eta \\
&\quad\leq
\left(
\int
\big(L\|\bx_2\|+S_0\big)^{2p}
\,\ud\eta
\right)^{1/2}
\left(
\int \|\bx_1-\bx_2\|^{2p}\,\ud\eta
\right)^{1/2} \\
&\quad=
\left(
\int_{\Real^N}
\big(L\|\bx\|+S_0\big)^{2p}
\,\ud\mu'(\bx)
\right)^{1/2}
W_{2p}^p(\mu,\mu').
\end{aligned}
\end{align}
Substituting \cref{eq:holder_first,eq:holder_second} into
\cref{eq:T_continuity_pre_holder}, we obtain
\begin{align*}
W_p^p\big(\calT(\mu),\calT(\mu')\big)
&\leq
C(\mu,\mu') W_{2p}^p(\mu,\mu'),
\end{align*}
where
\begin{align}
\label{eq:C_mu_mu_prime}
\begin{aligned}
C(\mu,\mu'):&=
2^{p-1}L^p W_{2p}^p(\mu,\delta_{\mathbf 0})\\
&\quad+
2^{p-1}
\left(
\int_{\Real^N}
\big(L\|\bx\|+S_0\big)^{2p}
\,\ud\mu'(\bx)
\right)^{1/2}.  
\end{aligned}
\end{align}
Now fix $\mu\in \calP_{2p}(\Real^N)$, and let $r > 0$ be arbitrary. If $W_{2p}(\mu,\mu')\leq r$, then by the triangle inequality, we have
\begin{align*}
W_{2p}(\mu',\delta_{\mathbf 0})
&\leq
W_{2p}(\mu',\mu)+W_{2p}(\mu,\delta_{\mathbf 0})
\leq
r+W_{2p}(\mu,\delta_{\mathbf 0}).
\end{align*}
Hence the $2p$-th moments of all $\mu'$ satisfying
$W_{2p}(\mu,\mu')\leq r$ are uniformly bounded. In particular,
\begin{align*}
\left(
\int_{\Real^N}
\big(L\|\bx\|+S_0\big)^{2p}
\,\ud\mu'(\bx)
\right)^{1/2}
\end{align*}
is uniformly bounded over the $W_{2p}$-ball
\begin{align*}
\{\mu'\in\calP_{2p}(\Real^N): W_{2p}(\mu,\mu')\leq r\}.
\end{align*}
Therefore, there exists a finite constant $C_{\mu,r}>0$, depending only on
$\mu$, $r$, $L$, $S_0$, and $p$, such that whenever
$W_{2p}(\mu,\mu')\leq r$,
\begin{align*}
W_p^p\big(\calT(\mu),\calT(\mu')\big)
&\leq
C_{\mu,r} W_{2p}^p(\mu,\mu').
\end{align*}
Equivalently,
\begin{align*}
W_p\big(\calT(\mu),\calT(\mu')\big)
&\leq
C_{\mu,r}^{1/p}W_{2p}(\mu,\mu').
\end{align*}
Now, given any $\varepsilon>0$, choose
\begin{align*}
\delta
&:=
\min\left\{
r,\frac{\varepsilon}{C_{\mu,r}^{1/p}}
\right\}.
\end{align*}
Then, whenever $W_{2p}(\mu,\mu')<\delta$, we have
\begin{align*}
W_p\big(\calT(\mu),\calT(\mu')\big)&<\varepsilon.
\end{align*}
Thus, $\calT$ is continuous at $\mu$ from the $W_{2p}$ topology to the $W_p$ topology. Since $\mu\in\calP_{2p}(\Real^N)$ is arbitrary, the proof is complete.

% \bibliographystyle{IEEEtran}
% \bibliography{bib/IEEEabrv,bib/StringDefinitions,bib/allref,bib/SIGNAL}

\end{document}